\documentclass{elsarticle}
\usepackage{fullpage}
\usepackage{url}
\usepackage{amsmath}
\usepackage{amssymb}
\usepackage{mathabx}
\usepackage{algorithm}
\usepackage{algpseudocode}
\usepackage{comment}
\usepackage{alltt}
\usepackage{nicefrac}

\usepackage{amsthm}
\theoremstyle{definition}
\newtheorem{definition}{Definition}
\newtheorem{theorem}{Theorem}
\newtheorem{example}{Example}
\newtheorem{lemma}{Lemma}
\newtheorem{proposition}{Proposition}
\newtheorem{corollary}{Corollary}
\theoremstyle{remark}
\newtheorem{remark}{Remark}

\usepackage{amssymb}

\newcommand{\union}{\cup}

\newcommand{\Nat}{\mathbb{N}}

\newcommand{\Var}{\mathrm{V}}	
\newcommand{\Terms}{\mathcal{T}}	
\newcommand{\sop}{[}
\newcommand{\scl}{]}
\newcommand{\sel}[2]{#1 \backslash #2}
\newcommand{\unsubst}[2]{\sop \sel{#1}{#2} \scl}

\newcommand{\impl}{\supset}
\newcommand{\Land}{\bigwedge}

\newcommand{\entails}{\vDash}

\newcommand{\Lang}{L}	
\newcommand{\Gram}{G}	

\usepackage{proof}
\newcommand{\LK}{\ensuremath{\mathbf{LK}}}	
\newcommand{\seq}{\rightarrow}	
\newcommand{\com}[1]{|{#1}|}	
\newcommand{\Hseq}{H}	
\newcommand{\cred}{\rightsquigarrow}	
\newcommand{\credne}{\stackrel{\mathit{ne}}{\rightsquigarrow}}	

\newcommand{\cut}{\mathrm{cut}}
\newcommand{\negl}{\neg_\mathrm{l}}

\newcommand{\andl}{\land_{\mathrm{l}}}

\newcommand{\impll}{\impl_\mathrm{l}}

\newcommand{\foralll}{\forall_\mathrm{l}}
\newcommand{\forallr}{\forall_\mathrm{r}}
\newcommand{\existsl}{\exists_\mathrm{l}}
\newcommand{\existsr}{\exists_\mathrm{r}}

\newcommand{\lkrule}[3][]{\infer[{#1}]{#2}{#3}}

\newcommand{\lkcut}[3][]{\infer[\cut{#1}]{#2}{#3}}
\newcommand{\lknl}[3][]{\infer[\negl{#1}]{#2}{#3}}

\newcommand{\lkil}[3][]{\infer[\impll{#1}]{#2}{#3}}

\newcommand{\lkur}[3][]{\infer[\forallr{#1}]{#2}{#3}}

\newcommand{\all}{\forall}
\newcommand{\DA}{{\rm GC}}
\newcommand{\GG}{{\rm GG}}
\newcommand{\ass}{\setminus}

\newcommand{\bAND}{\bigwedge}

\newcommand{\PCA}{{\rm PCA}} 
\newcommand{\cutintro}[2]{{\rm CI}(#1,#2)}

\newcommand{\nmodels}{\ensuremath{\nvDash}}
 
\newcommand{\FResOp}{\ensuremath{\mathcal{F}}} 

\newcommand{\dual}{\overline}
\newcommand{\res}{\ensuremath{\mathrm{res}}}
\newcommand{\comp}[1]{\circ_{#1}}

\newcommand{\mS}{\mathcal{S}}
\newcommand{\qequiv}{\quad\equiv\quad}
\newcommand{\Cons}{{\cal C}}
\newcommand{\SF}[1]{\ensuremath{\mathrm{SF}_{#1}}}    
\newcommand{\SFn}[1]{\ensuremath{\SF{#1}^+}}          
\newcommand{\CI}{\SFn{\FResOp}}                       
\newcommand{\dedClos}{\mathcal{D}}                  
\newcommand{\dedClosOp}{\ensuremath{\mS}}  
\newcommand{\SHS}{\ensuremath{\mathrm{SHS}}}

\def\RCS$#1: #2 ${\expandafter\def\csname RCS#1\endcsname{#2}}

\begin{document}

\begin{frontmatter}

\title{Algorithmic Introduction of Quantified Cuts\tnoteref{proj}}
\tnotetext[proj]{This work was supported by the projects P22028-N13, I603-N18 and
P25160-N25 of the Austrian Science Fund (FWF), by the ERC Advanced Grant
ProofCert and the WWTF Vienna Research Group 12-04.}

\author[dm]{Stefan Hetzl}
\ead{stefan.hetzl@tuwien.ac.at}
\author[tal]{Alexander Leitsch}
\ead{leitsch@logic.at}
\author[tal]{Giselle Reis}
\ead{giselle@logic.at}
\author[dm]{Daniel Weller}
\ead{weller@logic.at}

\address[dm]{Institute of Discrete Mathematics and Geometry, Vienna University of Technology, Vienna, Austria}
\address[tal]{Institute of Computer Languages, Vienna University of Technology, Vienna, Austria}

\begin{abstract}
We describe a method for inverting Gentzen's cut-elimination
in classical first-order logic. Our algorithm
is based on first computing a compressed representation
of the terms present in the cut-free proof and
then cut-formulas that realize such a compression.
Finally, a proof using these cut-formulas is constructed.
This method allows an exponential compression of proof length.
It can be applied to the output of automated theorem provers,
which typically produce analytic proofs. An implementation is
available on the web and described in this paper.
\end{abstract}

\end{frontmatter}


\section{Introduction}

Cut-elimination introduced by Gentzen~\cite{Gentzen34Untersuchungen} is the most
prominent form of proof transformation in logic and plays an important role in automating the
analysis of mathematical proofs. The removal of cuts corresponds to the elimination of intermediate statements (lemmas), resulting in a proof which is analytic in the
sense that all statements in the proof are subformulas of the result. Thus a proof of a combinatorial statement is converted into a purely combinatorial
proof. Cut-elimination is therefore an essential tool for the analysis of proofs, especially to make implicit parameters explicit.

In this paper we present a method for inverting Gentzen's cut-elimination by
computing a proof with cut from a given cut-free proof as input. As
cut-elimination is the backbone of proof theory, there is considerable proof-theoretic
interest and challenge in understanding this transformation sufficiently well to be able
to invert it. But our interest in cut-introduction is not only of a purely theoretical nature.
Proofs with cuts have properties that are essential for applications: one the
one hand, cuts are indispensable for formalizing proofs in a human-readable way.
One the other hand cuts have a very strong compression power in terms of proof length.

Computer-generated proofs are typically analytic, i.e.\ they only contain logical material that also appears in the theorem shown. This is due to the fact that
analytic proof systems have a considerably smaller search space which makes proof-search practically feasible. In the case of the sequent calculus, proof-search
procedures typically work on the cut-free fragment. But also resolution is essentially analytic as resolution proofs satisfy the subformula property of first-order
logic. An important property of non-analytic proofs is their considerably smaller length. The exact difference depends on the logic (or theory) under
consideration, but it is typically enormous. In (classical and intuitionistic) first-order logic there are proofs with cut of length $n$ whose theorems have only
cut-free proofs of length $2_n$ (where $2_0 = 1$ and $2_{n+1}=2^{2_n}$) (see~\cite{Statman79Lower} and~\cite{Orevkov79Lower}). The length of a proof plays an important
role in many situations such as human readability, space requirements and time requirements for proof checking. For most of these situations general-purpose data
compression methods cannot be used as the compressed representation is not a proof anymore. It is therefore of high practical interest to develop proof-search methods
which produce non-analytic and hence potentially much shorter proofs. In the method presented in this paper we start with a cut-free proof and abbreviate it by
computing useful cuts based on a structural analysis of the cut-free proof.

There is another, more theoretical, motivation for introducing cuts which derives from the foundations of mathematics: most of the central mathematical notions have
developed from the observation that many proofs share common structures and steps of reasoning. Encapsulating those leads to a new abstract notion, like that of a
group or a vector space. Such a notion then builds the base for a whole new theory whose importance stems from the pervasiveness of its basic notions in mathematics.
From a logical point of view this corresponds to the introduction of cuts into an existing proof database. While we cannot claim to contribute much to the
understanding of such complex historical processes by the current technical state of the art, this second motivation is still worthwhile to keep in mind, if only to
remind us that we are dealing with a difficult problem here.

Gentzen's method of cut-elimination is based on reductions of cut-derivations (subproofs ending in a cut), transforming them into simpler ones;
basically the cut is replaced by one or more cuts with lower logical complexity. A naive reversal of this procedure
is infeasible as it would lead to a search tree which is exponentially branching
on some nodes and infinitely branching on others. Therefore we base our procedure
on a deeper proof-theoretic analysis: in the construction of a Herbrand sequent $S'$ corresponding to a cut-free proof $\varphi'$ (see
e.g.~\cite{Baaz.Leitsch.1994}) obtained by cut-elimination on a proof $\varphi$ of a sequent $S$ with cuts, only the {\em substitutions} generated by cut-elimination
on quantified cuts are relevant. In fact, it is shown in~\cite{Hetzl12Applying} that, for proofs with $\Sigma_1$ and $\Pi_1$-cuts only, $S'$ can be obtained just
by computing
the substitutions defined by cut-elimination without applying Gentzen's procedure as a whole. Via the cuts in the proof $\varphi$ one can define a tree
grammar generating a language consisting exactly of the terms (to be instantiated for quantified variables in $S$) for obtaining the Herbrand sequent
$S'$~\cite{Hetzl12Applying}. Hence, generating a tree grammar $G$ from a set of Herbrand terms $T$ (generating $T$) corresponds to an inversion of the
quantifier part of Gentzen's procedure. The computation of such an inversion forms the basis of the method of {\em cut-introduction} presented in this paper.
Such an inversion of the quantifier part of cut-elimination determines which
instances of the cut-formulas are used but it does not determine the cut-formulas.
In fact, a priori it is not clear that every such grammar can be realized by
actual cut-formulas. However, we could show that, for any such tree
grammar representing the quantifier part of potential cut-formulas, actual
cut-formulas can be constructed. Finally, a proof containing these cut-formulas can be constructed.

Work on cut-introduction can be found at a number of different places in the literature.
Closest to our work are other approaches which aim to abbreviate or structure
{\em a given input proof}:~\cite{WoltzenlogelPaleo10Atomic} is an
algorithm for the introduction of atomic cuts that is capable of exponential proof compression.
The method~\cite{Finger07Equal} for propositional logic is
shown to never increase the size of proofs more than polynomially.
Another approach to the compression of first-order proofs by
introduction of definitions for abbreviating terms is~\cite{Vyskocil10Automated}.

Viewed from a broader perspective, this paper should be considered part of a large
body of work on the generation of non-analytic formulas that has been carried out
by numerous researchers in various communities. Methods for lemma generation are of crucial importance
in inductive theorem proving which frequently requires generalization~\cite{Bundy01Automation},
see e.g.~\cite{Ireland96Productive} for a method in the context of rippling~\cite{Bundy05Rippling}
which is based on failed proof attempts.
In automated theory formation~\cite{Colton01Automated,Colton02Automated}, an eager
approach to lemma generation is adopted. This work has, for example, led to
automated classification results 
of isomorphism classes~\cite{Sorge08Classification} and isotopy classes~\cite{Sorge08Automatic}
in finite algebra. See also~\cite{Johansson11Conjecture} for an approach to inductive theory
formation.
In pure proof theory, an important related topic is Kreisel's conjecture~(see footnote 3 on page 400 of~\cite{Takeuti87Proof})
on the generalization of proofs. Based on methods developed in this tradition,~\cite{Baaz92Algorithmic}
describes an approach to cut-introduction by filling a proof skeleton, i.e.\ an abstract proof structure,
obtained by an inversion of Gentzen's procedure with formulas in order
to obtain a proof with cuts.
The use of cuts for structuring and abbreviating proofs is also of relevance
in logic programming:~\cite{Miller07Tables} shows how to use focusing
in order to avoid proving atomic subgoals twice, resulting in a proof
with atomic cuts.

This paper is organized as follows:\\
In Section \ref{sec.pth_instruct} we define Herbrand sequents and {\em extended} Herbrand sequents which represent proofs with cut. The concept of rigid acyclic regular tree grammars is
applied to establish a relation between an extended Herbrand sequent $S^*$ and a (corresponding) Herbrand sequent $S'$: the language defined by this
grammar is just the set of terms $T$ to be instantiated for quantifiers in the original sequent $S$ to obtain $S'$. Given such a grammar $G$ generating $T$ there exists
a so-called {\em schematic extended Herbrand sequent} $\hat{S}$ in which the (unknown) cut-formulas are represented by monadic second-order variables. It is proved
that $\hat{S}$ always has a solution, the {\em canonical solution}. From this solution, which gives an extended Herbrand sequent and the cut-formulas for a proof, the
actual proof with these cuts is constructed.

To make the underlying methods more transparent, Section \ref{sec.pth_instruct} deals only with end-sequents of the form $\forall x\;F$.
In Section \ref{sec:more_general} the method is generalized to sequents of the form
$$\forall \bar{x}_1 F_1, \ldots, \forall \bar{x}_n F_n \seq \exists \bar{y}_1 G_1, \ldots, \exists \bar{y}_m G_m$$
where the $\bar{x}_i,\bar{y}_j$ are vectors of variables and $\forall z_1 \cdots z_k$ stands for $\forall z_1 \cdots \forall z_k$.
This form of sequents is more useful for practical applications and covers all of
first-order logic as an arbitrary sequent can be transformed into one of this
form by Skolemization and prenexification. We prove that all results obtained
in Section \ref{sec.pth_instruct} carry over to this more general case.

In Section \ref{sec:computation_of_grammar} an algorithm is presented computing a minimal rigid acyclic regular tree grammar generating the Herbrand term set $T$.

Given a cut-free proof $\varphi$ and a corresponding Herbrand term set $T$, the canonical solution corresponding to a non-trivial minimal grammar generating $T$ yields
a proof $\psi$ with lower quantifier-complexity (which is the number of quantifier inferences in a proof) than $\varphi$, but the length of $\psi$ (the total number
of inferences) may be greater than that of $\varphi$. In Section
\ref{sec.improving} a method is presented to overcome this problem. By using a resolution-based method the cut-formulas
are simplified under preservation of the quantifier complexity, resulting in proofs with lower number of inferences.

In Section \ref{sec.method_CI} a nondeterministic algorithm CI is defined which
is based on the techniques developed in Sections \ref{sec.pth_instruct} to
\ref{sec.improving}. We show that CI is, in a suitable sense, an inversion of
Gentzen's cut-elimination method.
A sequence of cut-free proofs is defined and it is proven
that the application of CI to this sequence results in an exponential compression of
proof length. Finally the existing implementation of CI and some experiments
are described in Section \ref{sec.implementation_experiments}.

This paper improves the publication~\cite{Hetzl12Towards} in several crucial
directions: (1) the method for introducing a single $\forall$-cut is generalized
to a method introducing an arbitrary number of $\forall$-cuts, which requires --
among others -- a length-preserving transformation of extended
Herbrand-sequents to proofs with cuts based on Craig interpolation, (2) we show
that our method is, in a suitable sense, an inversion of Gentzen's
cut-elimination method, (3) the end-sequent may contain blocks of quantifiers
instead of just single ones, (4) the decomposition of terms is represented as a
problem of grammars and a practical algorithm for computing grammars is
developed, (5) it is shown that the proof compression obtained by the new method is
exponential while it was only quadratic for the one of~\cite{Hetzl12Towards},
and (6) the algorithm has been
implemented in the gapt-system\footnote{\url{http://www.logic.at/gapt/}}.

The method CI developed in this paper is a systematic, proof-theoretic method to compress the lengths of first-order proofs by the introduction of cuts. Still, the
generated cut-formulas are all universal. A desirable extension of this method to introduce cuts with alternating quantifiers (which is necessary to obtain
super-exponential compressions) is left to future work.  Such an extension is highly non-trivial
as it first requires the development of an adequate notion of tree grammar
for extending the underlying proof-theoretic results to cuts with quantifier alternations.

\section{A Motivating Example}
\label{sec.motex}

Consider the sequents
\[
S_n\ =\  Pa, \forall x\, (Px\impl Pfx) \seq Pf^{2^n}a.
\]
The straightforward cut-free proof of $S_n$ in the sequent calculus uses the successor-axiom $2^n$ times. In fact, it is
easy to show that every cut-free proof has to contain all of these $2^n$ instances.
On the other hand, if we allow the use of cuts, we can give a considerably shorter
proof by first showing
\[
\forall x\, (Px\impl Pf^2x)
\]
from the axiom and then using this formula twice to show
\[
\forall x\, (Px\impl Pf^4x)
\]
and so on. In general we cut with the constant-length proofs of
\[
\forall x\, (Px\impl Pf^{2^i}x) \seq \forall x\, (Px\impl P^{2^{i+1}}x)
\]
and hence obtain a proof of $S_n$ that uses only $O(n)$ inference steps instead
of the $\Omega(2^n)$ of the cut-free proof. Note how
the structures of these two proofs are reminiscent of the binary and
the unary representation of numbers. This proof sequence is an exponential version
of the sequences of Statman~\cite{Statman79Lower} and Orevkov~\cite{Orevkov79Lower}
and has also been considered by Boolos~\cite{Boolos84Dont}.

In this paper we want to leverage this compression power of lemmas by automatically transforming
cut-free proofs into proofs using compressing lemmas.

\section{Proof-Theoretic Infrastructure}
\label{sec.pth_instruct}

Gentzen's proof of cut-elimination~\cite{Gentzen34Untersuchungen} can be understood
as the application of a set of local proof rewriting rules with a terminating strategy.
A first naive approach to cut-introduction would be to consider the inversion
of these local rewriting steps as a search algorithm. While this procedure would
in theory allow to reverse every cut-elimination sequence, it becomes clear quickly
that it is not feasible in practice: not only would we have to guess an enormous
amount of trivialities (e.g.\ rule permutations) but the inversion of rewriting
rules which erase a part of the proof lead to the necessity of correctly guessing
an entire subproof. Therefore we need more abstract proof representations.

The proof representations we will be using and their relationships are depicted
in Figure~\ref{fig.commdiag}.
\begin{figure}
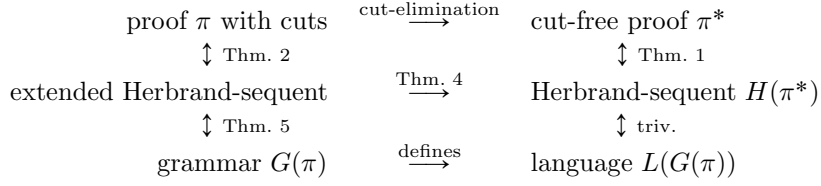

\begin{center}
\begin{tabular}{rcl}
proof $\pi$ with cuts & $\stackrel{\mbox{\scriptsize{cut-elimination}}}{\longrightarrow}$ & cut-free proof
$\pi^*$\\
$\updownarrow$ {\scriptsize Thm.~\ref{thm.proof_extHseq}}\hspace*{5mm}  & & \hspace*{11.3mm}$\updownarrow$ {\scriptsize Thm.~\ref{thm.proof_Hseq}}\\
extended Herbrand-sequent & $\stackrel{\mbox{\scriptsize{Thm.~\ref{thm.grammar_cutelim}}}}{\longrightarrow}$ &
Herbrand-sequent $\Hseq(\pi^*)$\\
$\updownarrow$ {\scriptsize Thm.~\ref{thm.extHseq_grammar}}\hspace*{5mm}  & & \hspace*{10mm} $\updownarrow$ {\scriptsize triv.} \\
grammar $\Gram(\pi)$ & $\stackrel{\mbox{\scriptsize{defines}}}{\longrightarrow}$ & language $\Lang(\Gram(\pi))$ 
\end{tabular}
\end{center}
\caption{Proof-theoretic setting of this paper}
\label{fig.commdiag}
\end{figure}
The purpose of this section is to explain these representations and their relationships.
As a first orientation let us just mention that the rows of
Figure~\ref{fig.commdiag} contain notions on increasingly abstract levels:
the level of proofs in the first row, that of formulas in the second row and
that of terms in the third row. The transformations in each
column are complexity-preserving (in a sense that will be made precise soon). The
transformation of an object in the left column to an object in the right column 
increases its complexity considerably (exponentially in this paper).

For this whole section, fix a quantifier-free formula $F$ with one free variable
$x$ s.t.\ $\forall x\, F$ is unsatisfiable. We will, for the sake of simplicity,
explain our algorithm first in the setting of proofs of the end-sequent $\forall x\, F\seq$.
We will show how to abbreviate a given cut-free proof of this sequent by the
introduction of cuts, which are of the form $\forall x\, A$ for $A$ quantifier-free,
such cuts will be called $\Pi_1$-cuts in the sequel. The algorithm will
then be generalised to less restrictive end-sequents in
Section~\ref{sec:more_general}.

\subsection{Proofs and Herbrand's Theorem}
\label{sec.proofs_Hthm}

A {\em sequent} is an ordered pair of sets of formulas, written as $\Gamma \seq \Delta$. While
the concrete variant of the sequent calculus is of little importance to the
algorithms presented in this paper let us, for the sake of precision, fix it
to be $\mathbf{G3c}+\mathrm{Cut_{cs}}$\footnote{$\mathbf{G3c}+\mathrm{Cut_{cs}}$
has no structural rules and all its rules are invertible.}
from~\cite{Troelstra2000}.
\begin{definition}[Herbrand-sequent]
A tautological sequent of the form $H:F\unsubst{x}{t_1},\ldots,F\unsubst{x}{t_n}\seq$ is
called a {\em Herbrand-sequent} of $\forall x\, F\seq$. We define $|H|=n$ and call
it the \emph{complexity of H}.
\end{definition}
%
We thus measure the number of instances
of $\forall x\, F$ used for showing its unsatisfiability. This complexity-measure
is of fundamental importance as the undecidability of first-order logic hinges on
it: a bound gives a decision procedure as most general unification can be used
for bounding the term size in the number of instances; it then only remains to enumerate
all possible instances having at most this bounding size. On the level of proofs we keep track of the number of used instances
by counting the number of $\foralll$- and $\existsr$-inferences, for the other quantifier
rules ($\forallr$ and $\existsl$) one application per formula suffices.
\begin{definition}
We define the {\em quantifier-complexity of a proof} $\pi$, written as $|\pi|_\mathrm{q}$ as the number of $\foralll$- and $\existsr$-inferences in $\pi$.
\end{definition}
Herbrand-sequents then correspond to cut-free proofs in the following sense.
\begin{theorem}\label{thm.proof_Hseq}
$\forall x\, F\seq $ has a cut-free
proof $\pi$ with $|\pi|_\mathrm{q} = l$
iff it has a Herbrand-sequent $H$ with $|H| = l$.
\end{theorem}
\begin{proof}[Proof Sketch]
Given $\pi$ we obtain $H$ by reading off the instances of $\forall x\, F$ from
the proof $\pi$ and collecting them in a sequent (if $\pi$ contains some duplicate
instances we add dummy instances to $H$ for obtaining $|H| = |\pi|_\mathrm{q}$).

Given $H$ we first compute any propositional proof of $H$ and obtain a cut-free
proof $\pi$ of $\forall x\, F\seq$ by introducing the universal quantifier for each of those instances
and applying a sufficient number of contractions.
\end{proof}
The above theorem shows that we can think of a Herbrand-sequent as a concise
representation of a cut-free proof. 
A first important step towards our cut-introduction
algorithm will be the generalisation of this relation to proofs with an arbitrary
number of $\Pi_1$-cuts (in a way similar to~\cite{Hetzl2009}).
%
\begin{definition}
\label{def.extHseq}
Let $u_1,\ldots,u_m$ be terms, let $A_1,\ldots,A_n$
be quantifier-free formulas, let $\alpha_1,\ldots,\alpha_n$ be variables, let
$V(t)$ denote the set of variables occurring in the term $t$, and let
$s_{i,j}$ for $1\leq i \leq n, 1\leq j \leq k_j$ be terms s.t.\ 
\begin{enumerate}
\item $\Var(A_{i}) \subseteq \{ \alpha_i,\ldots,\alpha_n \}$ for all $i$, and
\item\label{def.extHseq.varcondsij} $\Var(s_{i,j}) \subseteq \{ \alpha_{i+1},\ldots,\alpha_n \}$ for all $i,j$.
\end{enumerate}
Then the sequent
\[
H\ =\ F\unsubst{x}{u_1},\ldots,F\unsubst{x}{u_m},A_1 \impl \Land_{j=1}^{k_1} A_1\unsubst{\alpha_1}{s_{1,j}},\ldots, 
A_n \impl \Land_{j=1}^{k_n} A_n\unsubst{\alpha_n}{s_{n,j}} \seq
\]
is called an {\em extended Herbrand-sequent} of $\forall x\, F\seq$ if $H$ is
a tautology.
\end{definition}
What is this cryptic definition supposed to mean? An extended Herbrand-sequent of
the above form will represent a proof with $n$ $\Pi_1$-cuts whose cut formulas are
$\forall \alpha_1\, A_1,\ldots,\forall \alpha_n\, A_n$ (or sometimes minor variants thereof),
the $\alpha_i$ are the eigenvariables
of the universal quantifiers in these cut-formulas, the $s_{i,j}$ the terms
of the instances of the cut-formulas on the right-hand side of the cut and the $u_i$ the terms of the instances
of our end-formula $\forall x\, F$.
The complexity of an extended Herbrand-sequent $H$ of the above form 
is defined as $|H| = m + \sum_{j=1}^n k_j$. One can view an extended Herbrand-sequent
together with a propositional proof of it as a particular form of proof in
the $\varepsilon$-calculus~\cite{Hilbert39Grundlagen2} with the cuts corresponding
to the critical formulas.
We obtain the following correspondence to the sequent calculus:
\begin{theorem}\label{thm.proof_extHseq}
$\forall x\, F\seq$ has
a proof $\pi$ with 
$\Pi_1$-cuts and $|\pi|_\mathrm{q} = l$ iff it has
an extended
Herbrand-sequent $H$ with $|H|=l$. 
\end{theorem}
While this theorem looks plausible it is not as straightforward to prove as one
may expect. Its proof relies on Craig's interpolation theorem~\cite{Craig57Three}
which we briefly repeat here for the reader's convenience in the version of~\cite{Takeuti87Proof}
and restricted to propositional logic. We split a sequent into two parts
by writing it as a {\em partition} $\Gamma_1\ ;\ \Gamma_2 \seq \Delta_1\ ;\ \Delta_2$.
The purpose of doing so is merely to mark $\Gamma_1,\Delta_1$ as belonging to
one and $\Gamma_2,\Delta_2$ as belonging to the other part of the partition. The
logical meaning of $\Gamma_1\ ;\ \Gamma_2 \seq \Delta_1\ ;\ \Delta_2$ is 
just $\Gamma_1, \Gamma_2 \seq \Delta_1, \Delta_2$.
\begin{theorem}\label{thm.interpolation}
If a quantifier-free sequent $\Gamma_1\ ;\ \Gamma_2 \seq \Delta_1\ ;\ \Delta_2$
is a tautology, then there is a quantifier-free formula $I$ s.t.\ 
\begin{enumerate}
\item Both $\Gamma_1 \seq \Delta_1, I$ and $I, \Gamma_2\seq\Delta_2$ are tautologies, and
\item All atoms that appear in $I$ appear in both $\Gamma_1 \seq \Delta_1$ and $\Gamma_2 \seq \Delta_2$.
\end{enumerate}
\end{theorem}
\begin{proof}
See~\cite{Takeuti87Proof}.
\end{proof}


\begin{proof}[Proof of Theorem~\ref{thm.proof_extHseq}]
For the left-to-right direction we proceed analogously to the cut-free case: by
passing through the proof $\pi$ and reading off the instances of quantified formulas (of both
the end-formula and the cuts) we obtain an extended Herbrand-sequent $H$ with $|H|\leq |\pi|_\mathrm{q}$ (which
can be padded with dummy instances if necessary in order to obtain $|H| = |\pi|_\mathrm{q}$).

For the right-to-left direction let
\[
H\ =\ F\unsubst{x}{u_1},\ldots,F\unsubst{x}{u_m},A_1 \impl \Land_{j=1}^{k_1} A_1\unsubst{\alpha_1}{s_{1,j}},\ldots, 
A_n \impl \Land_{j=1}^{k_n} A_n\unsubst{\alpha_n}{s_{n,j}} \seq
\]
be an extended Herbrand-sequent and let us begin by introducing some abbreviations. For
a set of terms $T$ and a formula $F$, write $F\unsubst{x}{T}$ for the set of
formulas $\{ F\unsubst{x}{t} \mid t\in T \}$. Abbreviate the  ``cut-implication''
$A_i \impl \Land_{j=1}^{k_i} A_i\unsubst{\alpha_i}{s_{i,j}}$ as
$\mathrm{CI}_i$ and let $U = \{ u_1,\ldots,u_m \}$. Then $H$ can be written more succinctly
as $F\unsubst{x}{U}, \mathrm{CI}_1,\ldots,\mathrm{CI}_n \seq$.

Let $U_i = \{ u\in U \mid \Var(u) \subseteq \{ \alpha_{i+1},\ldots,\alpha_n \}\}$ for
$i=0,\ldots,n$. First we will show that it suffices to find
quantifier-free formulas $A'_1,\ldots,A'_n$ s.t.\ the sequent
\[
H'\ =\ F\unsubst{x}{U},A'_1\impl \Land_{j=1}^{k_1} A'_1\unsubst{\alpha_1}{s_{1,j}},\ldots, 
A'_n \impl \Land_{j=1}^{k_n} A'_n\unsubst{\alpha_n}{s_{n,j}} \seq
\]
has a proof of the following {\em linear form}:
\[
\infer[\impl_l]{F\unsubst{x}{U},\mathrm{CI}'_1,\ldots,\mathrm{CI}'_n\seq}{
  \infer*{F\unsubst{x}{U},\mathrm{CI}'_1,\ldots,\mathrm{CI}'_{n-1}\seq A'_n}{
      \infer[\impl_l]{F\unsubst{x}{U},\mathrm{CI}'_1\seq A'_2, \ldots, A'_n}{
        \infer*{F\unsubst{x}{U}\seq A'_1, \ldots, A'_n}{}
        &
        \infer*{\Land_{j=1}^{k_1} A'_1\unsubst{\alpha_1}{s_{1,j}}, F\unsubst{x}{U_1}\seq A'_2, \ldots, A'_n}{}
      }
  }
  &
  \hspace*{-50pt}\infer*{\Land_{j=1}^{k_n} A'_n\unsubst{\alpha_n}{s_{n,j}}, F\unsubst{x}{U_n}\seq}{}
}
\]
where $\mathrm{CI}'_i$ abbreviates $A'_i \impl \Land_{j=1}^{k_i} A'_i\unsubst{\alpha_i}{s_{i,j}}$.
This suffices because in the above proof we can introduce cuts and quantifiers by
replacing a segment of the form
\[
\infer[\impl_l]{F\unsubst{x}{U}, \mathrm{CI}'_1,\ldots,\mathrm{CI}'_i \seq A'_{i+1},\ldots,A'_{n}}{
  F\unsubst{x}{U},\mathrm{CI}'_1,\ldots,\mathrm{CI}'_{i-1} \seq A'_i,\ldots,A'_{n}
  &
  \Land_{j=1}^{k_i} A'_i\unsubst{\alpha_i}{s_{i,j}},F\unsubst{x}{U_i} \seq A'_{i+1},\ldots,A'_{n}
}
\]
by
\[
\lkcut{F\unsubst{x}{U_i}, \forall x\, F \seq A'_{i+1},\ldots,A'_{n}}{
  \lkur{F\unsubst{x}{U_i}, \forall x\, F \seq \forall x\, A'_i\unsubst{\alpha_i}{x}, A'_{i+1},\ldots,A'_{n}}{
    \lkrule[\foralll^*]{F\unsubst{x}{U_i}, \forall x\, F \seq A'_i,\ldots,A'_{n}}{
      F\unsubst{x}{U_{i-1}}, \forall x\, F \seq A'_i,\ldots,A'_{n}
    }
  }
  &
  \lkrule[\foralll^*]{\forall x\, A'_i\unsubst{\alpha_i}{x}, F\unsubst{x}{U_i} \seq A'_{i+1},\ldots,A'_{n}}{
    A'_i\unsubst{\alpha_i}{s_{i,j}}_{j=1}^{k_i}, F\unsubst{x}{U_i} \seq A'_{i+1},\ldots,A'_{n}
  }
}
\]
and finishing the proof at its root by
\[
\lkrule[\foralll^*]{\forall x\, F\seq }{
  F\unsubst{x}{U_n}, \forall x\, F \seq
}.
\]
This transformation results in a proof whose number of $\foralll$-inferences is the
complexity of the extended Herbrand-sequent as every term of $H$ is introduced
exactly once.

Let us now turn to the construction of the $A'_i$. Write
\begin{align*}
L_i &\quad \mbox{for}\quad \mathrm{CI}_1,\ldots,\mathrm{CI}_{i-1}, F\unsubst{x}{U} \seq A'_i,\ldots,A'_n,\ \mbox{and}\\
R_i &\quad \mbox{for}\quad \Land_{j=1}^{k_i} A'_i\unsubst{\alpha_i}{s_{i,j}}, F\unsubst{x}{U_i} \seq A'_{i+1},\ldots,A'_{n}.
\end{align*}
Note that $L_i$ and $R_i$ depend only on those $A'_j$ with $j\geq i$ and note furthermore
that $L_{n+1}$ is the extended Herbrand-sequent $H$ which is a tautology by
assumption. Fix $i\in\{1,\ldots,n\}$. Assuming $\entails L_{i+1}$ we will
now construct $A'_i$ and show $\entails L_i$ and $\entails R_i$.

From $\entails L_{i+1}$ we obtain
\begin{align}
\entails & 
  \mathrm{CI}_1,\ldots,\mathrm{CI}_{i-1},F\unsubst{x}{U} \seq A_i, A'_{i+1},\ldots,A'_{n} \ \mbox{and}\label{eq.extHseq.lhs}\\
\entails & 
  \underbrace{\mathrm{CI}_1,\ldots,\mathrm{CI}_{i-1},F\unsubst{x}{(U \setminus U_i)}}_{\Gamma}, 
  \underbrace{\Land_{j=1}^{k_i} A_i\unsubst{\alpha_i}{s_{i,j}}, F\unsubst{x}{U_i}}_{\Pi} \seq 
  \underbrace{A'_{i+1}, \ldots, A'_{n}}_{\Lambda}\label{eq.extHseq.rhs}
\end{align}
from an application of $\impll$ to $\mathrm{CI}_i$. Applying the propositional
interpolation theorem to the partition $\Gamma \ ;\  \Pi \seq \ ;\  \Lambda$
of~(\ref{eq.extHseq.rhs}) yields $I$ s.t.\ $\entails \Gamma\seq I$ and $\entails \Pi, I \seq \Lambda$. Furthermore
$I$ contains only such atoms which appear in $\Pi \seq \Lambda$, hence
$\Var(I) \subseteq \{ \alpha_{i+1},\ldots,\alpha_n \}$. Define $A'_i$ as $A_i \land I$.
Observe that $\entails R_i$ follows from $\entails \Pi, I \seq \Lambda$ and
$\entails L_i$ follows from~(\ref{eq.extHseq.lhs}) and $\entails \Gamma \seq I$.
Hence $L_i,R_i$ for $i=1,\ldots,n$ are tautologies. But $L_1,R_1,\ldots,R_n$
are exactly the leaves of the linear proof from above which
finishes the proof of the theorem.
\end{proof}
This result does not only generalize Proposition~2 of~\cite{Hetzl12Towards} to
the case of an arbitrary number of cuts but also improves it considerably, even for
the case of a single cut: the use of interpolants is new in this paper and allows
to obtain $|\pi|_\mathrm{q} \leq |H|$ for an extended Herbrand sequent $H$. In
general it is not possible to read back an extended Herbrand-sequent to a proof
of linear form without changing the cut formulas as the following example shows.
The reason for insisting on this linear form is that it does not contain any
duplicate instances which permits to show the property $|\pi|_\mathrm{q} = |H|$.
The duplication behavior of connectives in this transformation is reminiscent
of the complexity results in~\cite{Baaz12Complexity}.
\begin{remark}
The complexity of the proof $\pi$ obtained from the extended Herbrand-sequent $H$
can also be bound beyond its pure quantifier complexity $|\pi|_\mathrm{q}$.
Let $\mathrm{d}(\psi)$ denote the depth of a proof $\psi$,
i.e.\ the maximal number of inferences on a branch and let $\| H \|$ denote the
logical complexity of $H$.
Then the right-to-left direction of Theorem~\ref{thm.proof_extHseq} can be
strengthened as follows: there is a constant $c$ s.t.\ for every extended
Herbrand sequent $H$ of $\forall x\, F\seq$ with $n$ cuts and $|H|=l$ there is a proof $\pi$
with $n$ $\Pi_1$-cuts, $|\pi|_\mathrm{q} = l$ and $\mathrm{d}(\pi) \leq c^n \| H \|$. This
bound can be obtained from carrying out the proofs of Theorem~\ref{thm.interpolation}
and Theorem~\ref{thm.proof_extHseq} using $\vdash^d$ (derivability in depth $d$)
instead of $\entails$ (validity). It is created by the $n$-fold iteration of transformations of
a proof of depth $d$ to a proof of depth $c\cdot d$.
\end{remark}
\begin{example}
Let $F = P(x) \land (P(c) \impl Q(x)) \land (Q(x)\impl P(d)) \land \neg P(d)$ and
$A_1 = P(\alpha_1)$. Furthermore let $m=1, u_1 = \alpha_1 $ and $n=1,k_1=1,s_{1,1} = c$. Then
\begin{align*}
E & = F\unsubst{x}{u_1},\ldots,F\unsubst{x}{u_m},A_1 \impl \Land_{j=1}^{k_1} A_1\unsubst{\alpha_1}{s_{1,j}},\ldots, 
A_n \impl \Land_{j=1}^{k_n} A_n\unsubst{\alpha_n}{s_{n,j}} \seq\\
& = P(\alpha_1) \land (P(c) \impl Q(\alpha_1)) \land (Q(\alpha_1)\impl P(d)) \land \neg P(d), P(\alpha_1)\impl P(c)\seq
\end{align*}
is a tautology and hence an extended Herbrand-sequent of $\forall x\, F\seq$.

Let us now try to construct a linear {\LK}-proof that corresponds to $E$.
Such a proof contains a cut on $\forall x\, P(x)$ as its last inference. The
formula $\forall x\, F$ must be instantiated on the left above this cut
to obtain $F\unsubst{x}{\alpha_1}$ as $\alpha_1$ is the eigenvariable
of the cut formula. This leaves the right side of the cut as $P(c)\seq$ which
is not valid, a second instance of $\forall x\, F$ would be needed. The solution
used in the proof of Theorem~\ref{thm.proof_extHseq} is based on computing 
a propositional interpolant of $F\unsubst{x}{\alpha_1} ; P(c) \seq ; $. This can be
done e.g.\ by first computing a proof of the sequent $F\unsubst{x}{\alpha_1}, P(c) \seq $,
e.g.\ the following $\psi = $
\[
\lkrule[\land_{\mathrm{l}}^*]{P(\alpha_1) \land (P(c) \impl Q(\alpha_1)) \land (Q(\alpha_1)\impl P(d)) \land \neg P(d), P(c) \seq }{
  \lkil{P(\alpha_1), P(c) \impl Q(\alpha_1), Q(\alpha_1)\impl P(d), \neg P(d), P(c) \seq }{
    P(c)\seq P(c)
    &
    \lkil{P(\alpha_1), Q(\alpha_1), Q(\alpha_1)\impl P(d), \neg P(d) \seq}{
      Q(\alpha_1)\seq Q(\alpha_1)
      &
      \lknl{P(\alpha_1), P(d), \neg P(d) \seq}{
        P(\alpha_1), P(d) \seq P(d)
      }
    }
  }
}
\]
The propositional interpolant induced by the partition $F\unsubst{x}{\alpha_1} ; P(c) \seq ; $ of
$\psi$ according to the algorithm of~\cite{Takeuti87Proof} is computed as
\[
\lkrule[\land_{\mathrm{l}}^*]{\neg P(c) \lor \bot \lor \bot}{
  \lkil{\neg P(c) \lor \bot \lor \bot}{
    \neg P(c)
    &
    \lkil{\bot\lor\bot}{
      \bot
      &
      \lknl{\bot}{
        \bot
      }
    }
  }
}
\]
which simplifies to $\neg P(c)$. Hence the new cut formula is $\forall x\, (P(x)\land \neg P(c))$
which renders the right
side of the cut provable as $P(c)\land\neg P(c)\seq$.
\end{example}

\subsection{Proofs and Grammars}

Now that we have established the connection between proofs and (extended)
Herbrand-sequents we can move on to the term level of Figure~\ref{fig.commdiag}. A first trivial observation is
that, assuming the knowledge of $F$, a Herbrand-sequent $H$ for $\forall x\, F\seq$
does not carry more information than just the set of terms $T$ s.t.\ $H = F\unsubst{x}{T} \seq$.

A set of terms, in the terminology of formal language theory, is a tree language.
The central theoretical result on which this paper is based
is an analogous relation between extended Herbrand-sequents (or: proofs with
$\Pi_1$-cuts) and {\em a certain class of tree grammars}. This result has first been
proved in~\cite{Hetzl12Applying}, see also~\cite{Hetzl12Herbrand} for a generalization.

Tree languages are a natural generalization of formal (string) languages, see
e.g.~\cite{Gecseg97Tree,Comon07Tree}. Many
important notions, such as regular and context-free languages carry over from the setting
of strings to that of trees. The class of {\em rigid} tree languages has been introduced
in~\cite{Jacquemard09Rigid} with applications in verification in mind, see~\cite{Jacquemard11Rigid}.
Rigid tree languages augment regular tree
languages by the ability to carry out certain equality tests, a property that
is very useful for applications.

In the context of proof theory it is more natural to work with grammars than
with automata because of the generative nature of cut-elimination.
The class of grammars we will use
in this paper is a subclass of rigid grammars: the {\em totally rigid acyclic
tree grammars}. We write $\Terms_\Sigma(V)$ for the set of terms in the
first-order signature $\Sigma$ over the set of variables $V$ and $\Terms_\Sigma$
for $\Terms_\Sigma(\emptyset)$. For a symbol $f \in \Sigma$ we write $(f/k)$ for
denoting the arity $k$ of $f$.
\begin{definition}
A {\em regular tree grammar} is a tuple $G = \langle N,\Sigma,\tau, P \rangle$,
where $N$ is a finite set of non-terminal symbols, $\Sigma$ is a first-order
signature, $\tau \in N$ is the {\em start symbol} and $P$ is a finite
set of production rules of the form $\beta\rightarrow t$ with $\beta \in N$ and 
$t\in\Terms_{\Sigma}(N)$.
\end{definition}
The {\em one-step derivation relation $\rightarrow^1_G$} of a regular tree grammar $G$ consists
of all pairs $u[\beta]\rightarrow^1_G u[t]$ where $\beta\rightarrow t \in P$. A derivation in $G$
is a finite sequence of terms $t_0 = \tau, t_1,\ldots,t_n$ s.t.\  $t_i \rightarrow^1_G t_{i+1}$.
The language of $G$ is defined as $\Lang(G) = \{ t\in\Terms_{\Sigma} \mid t\ \mbox{has a $G$-derivation} \}$.
\begin{definition}
A {\em rigid tree grammar} is a tuple $G = \langle N,N_R,\Sigma,\tau,P \rangle$, where $\langle N,\Sigma,\tau,P \rangle$, is a
regular tree grammar and $N_R\subseteq N$ is the set of \emph{rigid non-terminals}.
We speak of a \emph{totally rigid tree grammar} if $N_R=N$. In this case we will just write
$\langle N_R,\Sigma,\tau,P \rangle$.
\end{definition}
A derivation $t_0 = \tau, t_1, \ldots, t_n=t$ of a term $t\in\Terms_\Sigma$ in
a rigid tree grammar is a derivation in the underlying regular tree grammar that satisfies the additional
{\em rigidity condition}: If there are $i,j<n$, a non-terminal $\beta\in N_R$,
and positions $p$ and $q$ such that $t_i|_p=\beta$ and $t_j|_q=\beta$, then $t|_p=t|_q$.
The language $\Lang(G)$ of the rigid
tree grammar $G$ is the set of all terms $t\in\Terms_\Sigma$ which can
be derived under the rigidity condition. Totally rigid tree grammars are
formalisms for specifying sets of substitutions and thus are particularly
useful for describing instances generated by cut-elimination.
\begin{example}
Let $\Sigma=\{0/0,s/1\}$. A simple pumping argument shows that the language $L=\{ f(t,t) \mid t\in\Terms_\Sigma \}$
is not regular. On the other hand, $L$ is generated by the rigid tree grammar
$\langle \{ \tau,\alpha,\beta \}, \{ \alpha \} , \{ 0/0, s/1, f/2 \}, \tau, P\rangle$ where 
$
P = \{ \tau \rightarrow f(\alpha,\alpha), \alpha \rightarrow 0 \mid s(\beta), \beta \rightarrow 0 \mid s(\beta) \}$.
\end{example}
\begin{definition}\label{def.gram_extHseq}
The {\em grammar of an extended Herbrand-sequent}
\[
H\ \equiv\ F\unsubst{x}{u_1},\ldots,F\unsubst{x}{u_m},A_1 \impl \Land_{j=1}^{k_1} A_1\unsubst{\alpha_1}{s_{1,j}},\ldots, 
A_n \impl \Land_{j=1}^{k_n} A_n\unsubst{\alpha_n}{s_{n,j}} \seq
\]
is defined as the totally rigid $\Gram(H) = \langle N_R, \Sigma, \tau, P\rangle$ where $N_R = \{ \tau, \alpha_1,\ldots,\alpha_n \}$,
$\Sigma$ is the signature of $H$ and $P = \{ \tau \rightarrow u_i \mid 1\leq i \leq m\}
\union \{ \alpha_i \rightarrow s_{i,j} \mid 1 \leq i \leq n, 1\leq j \leq k_i \} $.
\end{definition}
A derivation of the form $\beta \rightarrow^1_G t_1 \rightarrow^1_G \cdots \rightarrow^1_G t_n$
is called cyclic if $\beta \in \Var(t_n)$. A grammar is called {\em acyclic} if it
does not have any cyclic derivations. Note that condition~\ref{def.extHseq.varcondsij} of Definition~\ref{def.extHseq}
ensures that the grammar of an extended Herbrand-sequent is acyclic. Furthermore,
by definition, the grammar of an extended Herbrand-sequent is totally rigid. The
language of such a grammar can be written in the following normal form.
\begin{lemma}\label{lem.totrig_acyclic_language}
If $G$ is totally rigid and acyclic, then up to renaming of the non-terminals
$G=\langle \{\alpha_0,\ldots,\alpha_n\}, \Sigma, \alpha_0, P\rangle$ with
$\Lang(G)=\{ \alpha_0\unsubst{\alpha_0}{t_0}\cdots\unsubst{\alpha_n}{t_n} \mid \alpha_i \rightarrow t_i \in P\}$.
\end{lemma}
\begin{proof}
Acyclicity permits to rename the non-terminals in such a way that $\alpha_i \rightarrow^1_G t_1 \rightarrow^1_G \cdots \rightarrow^1_G t_n$
and $\alpha_j \in \Var(t_n)$ implies $j > i$. The notation based on substitutions
is then possible because, due to total rigidity, each $t\in\Lang(G)$ can be derived using at most one
production for each non-terminal. See~\cite{Hetzl12Herbrand} for a detailed proof.
\end{proof}
In particular, the language of a totally rigid acyclic grammar is finite. This
lemma also suggests a compact notation for totally rigid acyclic grammars: we write
\[
U \circ_{\alpha_1} S_1 \cdots \circ_{\alpha_n} S_n
\]
for the grammar $\langle \{ \tau, \alpha_1,\ldots,\alpha_n \}, \Sigma, \tau, P\rangle$
where $P = \{ \tau \rightarrow u \mid u \in U \} \union \{ \alpha_i \rightarrow s_i \mid 1 \leq i \leq n, s_i \in S_i \}$,
$\tau$ is some fresh start symbol and $\Sigma$ is the signature of the terms appearing in $P$.
Using this notation, we can observe that $\Lang(U) = U$ for a set of terms $U$ and
$\Lang(G \circ_\alpha S) = \{ u\unsubst{\alpha}{s} \mid u \in \Lang(G), s\in S\}$
for a totally rigid acyclic tree grammar $G$ and a set of terms $S$.
If the non-terminals are clear from the context, this notation
is further abbreviated as
\[
U \circ S_1 \cdots \circ S_n.
\]
One can then obtain a cut-elimination theorem based on grammars:
\begin{theorem}\label{thm.grammar_cutelim}
If $H$ is
an extended Herbrand-sequent of $\forall x\, F\seq$, then 
$\{ F\unsubst{x}{t} \mid t\in\Lang(\Gram(H)) \}\seq$
is a Herbrand-sequent of $\forall x\, F\seq$.
\end{theorem}
\begin{proof}
This can be shown by following the development of the grammar during
a cut-elimination process, see~\cite{Hetzl12Applying,HetzlXXProofs}
and also~\cite{Hetzl12Herbrand} for a more general result.
\end{proof}
Throughout this whole paper all the grammars we are dealing with will be
totally rigid and acyclic. Therefore we will henceforth use {\em grammar}
as synonym for {\em totally rigid acyclic tree grammar}.

\subsection{Cut-Introduction}
\label{sec:cutintroduction}
We have already observed above that it is not feasible to invert Gentzen's cut-elimination
steps literally. The key to our method is that moving from the level of proofs to the level of grammars 
provides us with a transformation that is much easier to invert. The computation of the
language of a grammar can simply be inverted as: {\em given a finite tree language $L$,
find a grammar $G$ s.t.\ $\Lang(G) = L$}.
 We will describe
an algorithm for solving this problem in detail in Section~\ref{sec:computation_of_grammar}.

The only piece then still missing in Figure~\ref{fig.commdiag} is to obtain an extended Herbrand-sequent from $G$.
Note that, for a given $G$, the term-part of the extended Herbrand-sequent
is already determined using Lemma~\ref{lem.totrig_acyclic_language}. What we do not know yet are the cut-formulas.
Hence we define:
\begin{definition}\label{def.sch_ext_Hseq}
Let $u_1,\ldots,u_m$ be terms, let $X_1,\ldots,X_n$ be monadic second-order variables, let $\alpha_1,\ldots,\alpha_n$ be variables, and let
$s_{i,j}$ for $1\leq i \leq n, 1\leq j \leq k_j$ be terms
s.t.\ $\Var(s_{i,j}) \subseteq \{ \alpha_{i+1},\ldots,\alpha_n \}$ for all $i,j$.
Then the sequent
\[
H\ =\ F\unsubst{x}{u_1},\ldots,F\unsubst{x}{u_m},X_1(\alpha_1) \impl \Land_{j=1}^{k_1} X_1(s_{1,j}),\ldots, 
X_n(\alpha_n) \impl \Land_{j=1}^{k_n} X_n(s_{n,j}) \seq
\]
is called a {\em schematic extended Herbrand-sequent} of $\forall x\, F\seq$ if
$\Land_{t\in\Lang(\Gram(H))} F\unsubst{x}{t}\seq$ is a tautology (where $\Gram(H)$
is defined analogously to Definition~\ref{def.gram_extHseq}).

A {\em solution} of a schematic extended Herbrand-sequent $H$ is a substitution
$\sigma = \unsubst{X_i}{\lambda \alpha_i. A_i}_{i=1}^{n}$ s.t.\ $\Var(A_i) \subseteq \{ \alpha_{i},\ldots,\alpha_n\}$
and $H\sigma$ is a tautology.
\end{definition}
The reason for calling such a substitution $\sigma$ a solution is the close
relationship of this problem to unification problems modulo the theory
of Boolean algebras, in particular to Boolean unification with
constants~\cite{Martin89Boolean,Baader98Complexity}.
By comparison with Definition~\ref{def.extHseq} note that if $\sigma$ is
a solution for $H$, then $H\sigma$ is an extended Herbrand-sequent. There are
a number of interesting and practically relevant results about the solutions of such sequents, 
see Section~\ref{sec.improving}. The central property which is of interest right now is that
such a sequent always has a solution.
\begin{definition}\label{def.can_subst}
Let $H$ be a schematic extended Herbrand-sequent. Define
\[
C_1\ = \Land_{i=1}^m F\unsubst{x}{u_i}\ \mbox{and}\ 
C_{i+1}\ = \Land_{j=1}^{k_i} C_i\unsubst{\alpha_i}{s_{i,j}}\ \mbox{for}\ i=1,\ldots,n.
\]
Then
\[
\sigma := \unsubst{X_i}{\lambda \alpha_i. C_i}_{i=1}^n
\]
is called {\em canonical substitution} of $H$.
\end{definition}
We will now show that the canonical substitution is, in fact, a solution.
\begin{lemma}\label{lem.can_valid}
Let $H$ and $C_i$ be as in Definition~\ref{def.can_subst}. Then
$C_{n+1}\seq$ is a tautology.
\end{lemma}
\begin{proof}
By definition, $C_{n+1}\seq$ is $\Land_{i=1}^{m} \Land_{j_1=1}^{k_1} \cdots \Land_{j_n=1}^{k_n}
F\unsubst{x}{u_i}\unsubst{\alpha_1}{s_{1,j_1}}\cdots\unsubst{s_{n,\alpha_n}}{j_n} \seq$ which
by Lemma~\ref{lem.totrig_acyclic_language} is $\Land_{t\in\Lang(\Gram(H))} F\unsubst{x}{t} \seq$
which is a tautology as $H$ is a schematic extended Herbrand-sequent.
\end{proof}
\begin{lemma}\label{lem.can_sol}
Let $H$ be a schematic extended Herbrand-sequent and $\sigma$ be its canonical substitution.
Then $\sigma$ is a solution of $H$.
\end{lemma}
\begin{proof}
First note that the variable condition is fulfilled as $\Var(C_i) \subseteq \{ \alpha_i,\ldots, \alpha_n \}$.
Then observe that
\[
H\sigma\ =\ F\unsubst{x}{u_1},\ldots,F\unsubst{x}{u_m},C_1 \impl \Land_{j=1}^{k_1} C_1\unsubst{\alpha_1}{s_{1,j}},\ldots, 
C_n \impl \Land_{j=1}^{k_n} C_n\unsubst{\alpha_n}{s_{n,j}} \seq
\]
is logically equivalent to
\[
C_1, C_1 \impl C_2, \ldots, C_n \impl C_{n+1} \rightarrow.
\]
The unsatisfiability of $C_{n+1}$ follows from Lemma~\ref{lem.can_valid},
hence $H\sigma$ is a tautology.
\end{proof}
In light of the above result we will henceforth call $\sigma$ the {\em canonical solution}.
Note that the canonical solution permits a sequent calculus proof of a
linear form in the sense of the proof of Theorem~\ref{thm.proof_extHseq}, and hence
--- for this solution --- interpolation is not necessary in the construction of the proof
with cuts.
\begin{theorem}\label{thm.extHseq_grammar}
$\forall x\, F\seq$ has an extended Herbrand-sequent $H$ with
$|H| = l$ iff there is a totally rigid acyclic tree grammar $G$ with $|G|=l$ s.t.\ 
$\Land_{t\in\Lang(G)} F\unsubst{x}{t}\seq$ is a tautology.
\end{theorem}
\begin{proof}
The left-to-right direction of this statement follows from Theorem~\ref{thm.grammar_cutelim}
together with the observation that $|\Gram(H)| = |H|$. For the right-to-left direction
assume that $G$ is given, let $H$ be the schematic extended Herbrand-sequent of
$G$. Then the result follows from Lemma~\ref{lem.can_sol}.
\end{proof}
Now we have proved all results mentioned in Figure~\ref{fig.commdiag} and can finally
describe our approach to cut-introduction. It consists in
following this diagram in a clockwise fashion from the cut-free proof
to the proof with cut. More specifically, given as input a cut-free proof $\pi$
our algorithm will proceed as follows:
\begin{enumerate}
\item Extract the set of terms $T$ of $\Hseq(\pi)$ (as in Theorem~\ref{thm.proof_Hseq}).
\item Find a suitable grammar $G$ s.t.\ $\Lang(G) = T$.
\item Compute an extended Herbrand-sequent $H$ from $G$ (as in Theorem~\ref{thm.extHseq_grammar}).
\item Construct a proof $\psi$ with cut from $H$ (as in Theorem~\ref{thm.proof_extHseq}).
\end{enumerate}

\begin{example}
Consider the sequent $\forall x\, F \seq$ where
\[
F =  Pa \land (Px \impl Pfx) \land \neg Pf^9a.
\]
Let $\pi$ be a straightforward cut-free proof of $\forall x\, F\seq$, then $|\pi|_\mathrm{q} = 9$.
Following the above outline of an algorithm we carry out the following steps.
Extract the set of terms
\[
 T = \{ a, fa, f^2a, f^3a, f^4a, f^5a, f^6a, f^7a, f^8a \}
\]
from $\Hseq(\pi)$ following Theorem~\ref{thm.proof_Hseq}. Compute a grammar $G$ with $\Lang(G)=T$, for
example
\[
G = \{ \alpha, f\alpha, f^2\alpha \} \circ_{\alpha} \{ a, f^3a, f^6a \}.
\]
As in the proof of Theorem~\ref{thm.extHseq_grammar}, this grammar induces the schematic extended Herbrand-sequent
\[
H\ =\ F\unsubst{x}{\alpha}, F\unsubst{x}{f\alpha}, F\unsubst{x}{f^2\alpha},
X(\alpha) \impl (X(a) \land X(f^3a) \land X(f^6a)) \seq
\]
whose canonical solution is
\[
\sigma = \unsubst{X}{\lambda \alpha.\, ( F\unsubst{x}{\alpha} \land F\unsubst{x}{f\alpha} \land F\unsubst{x}{f^2\alpha} )}
\]
Hence $H\sigma$ is an extended Herbrand-sequent with $|H\sigma| = |H| = 6$ which
in turn induces a proof $\psi$ as in Theorem~\ref{thm.proof_extHseq} which
has $|\psi|_\mathrm{q} = 6$ and contains a single $\Pi_1$-cut whose cut-formula is
\[
 \forall x\, (F \land F\unsubst{x}{fx} \land F\unsubst{x}{f^2x} ).
\]
Observe that we have decreased the quantifier complexity from $|\pi|_\mathrm{q} = 9$
to $|\psi|_\mathrm{q} = 6$.
\end{example}
While this is a satisfactory situation from the abstract point of view of the quantifier
complexity, this procedure is clearly not yet fit for practical applications with
the aim of proof compression. The rest of this paper is devoted to making it so:
in Section~\ref{sec:more_general} we generalize the results of this section to
a sufficiently large class of end-sequents. In Section~\ref{sec:computation_of_grammar}
we present an efficient algorithm for the computation of a grammar and
in Section~\ref{sec.improving} we describe how to obtain solutions for
a schematic extended Herbrand-sequent which are shorter than the canonical solution.

\section{More General End-Sequents}
\label{sec:more_general}

The class of end-sequents considered in the previous section, while leading
to a comparatively simple statement of the central results, is clearly too
restricted for concrete applications. We will therefore extend our
proof-theoretic infrastructure to proofs of end-sequents of the form
\[
\forall x_1 \cdots \forall x_{l_1}\, F_1, \ldots, \forall x_1\cdots\forall x_{l_p} F_p \seq
\exists x_1\cdots\exists x_{l_{p+1}}\, F_{p+1},\ldots, \exists x_1\cdots\exists x_{l_q} F_q
\]
with $l_i\geq 0$ and $F_i$ quantifier-free. We say that a sequent in this format
is a $\Sigma_1$-sequent. Note that every first-order sequent
can be transformed to this form by skolemization and prenexing. Permitting $l_i$ to
be zero allows for quantifier-free formulas such as in the example of
Section~\ref{sec.motex}. While the formalism now gets notationally more complicated,
the results and proofs remain essentially the same. We write $\bar{x}$ for a
vector $(x_1,\ldots,x_n)$ of variables, $\bar{t}$ for a vector $(t_1,\ldots,t_n)$
of terms and $\unsubst{\bar{x}}{\bar{t}}$ for the substitution
$\sop \sel{x_1}{t_1},\ldots,\sel{x_n}{t_n} \scl$. For this whole section,
we fix a sequent $\Gamma\seq\Delta$ of the above form.
\begin{definition}\label{def.herbrand.sequent}
A tautological sequent of the form
\[
\{ F_i\unsubst{\bar{x}}{\widebar{t_{i,j}}} \mid  1 \leq i \leq p, 1\leq j \leq n_i\} \seq
\{ F_i\unsubst{\bar{x}}{\widebar{t_{i,j}}} \mid  p < i \leq q, 1\leq j \leq n_i \} 
\]
is called {\em Herbrand-sequent} of $\Gamma\seq\Delta$.
\end{definition}
The size of a Herbrand-sequent is defined as $|H| = \sum_{i=1}^{q} n_i$. Note
that we only count formulas obtained by instantiation. Now
as we are dealing with blocks of quantifiers it is appropriate to also change
the size measure on proofs to consider blocks instead of single quantifiers. To
that aim we change the quantifier rules in our sequent calculus to allow the
introduction of a block of quantifiers (which is a natural alternative for a number
of problems related to proof size, see e.g.~\cite{Baaz92Algorithmic}):
\[
\infer[\forall_\mathrm{l}^*]{\forall x_1\cdots \forall x_n\, A, \Gamma \seq \Delta}{
  \forall x_1\cdots \forall x_n\, A, A\unsubst{\bar{x}}{\bar{t}}, \Gamma \seq \Delta
}
\qquad
\infer[\exists_\mathrm{r}^*]{\Gamma\seq\Delta, \exists x_1\cdots \exists x_n\, A}{
  \Gamma\seq\Delta, A\unsubst{\bar{x}}{\bar{t}}, \exists x_1\cdots \exists x_n\, A
}
\]
We write $\| \pi \|_\mathrm{q}$ for the number of $\forall^*_\mathrm{l}$- and
$\exists^*_\mathrm{r}$-inferences in the proof $\pi$.
\begin{theorem}
\label{thm.proof_Hseq2}
$\Gamma\seq\Delta$ has a cut-free proof $\pi$ with $\|\pi\|_\mathrm{q}=l$ iff
it has a Herbrand-sequent $H$ with $|H|=l$.
\end{theorem}
\begin{proof}
As for Theorem~\ref{thm.proof_Hseq}.
\end{proof}

\begin{definition}
Let $\widebar{u_{i,1}},\ldots,\widebar{u_{i,m_i}}$ be vectors of terms with
$l_i$ elements each. Let $A_1,\ldots,A_n$
be quantifier-free formulas, let $\alpha_1,\ldots,\alpha_n$ be variables, and let
$s_{i,j}$ for $1\leq i \leq n, 1\leq j \leq k_j$ be terms s.t.\ 
\begin{enumerate}
\item $\Var(A_{i}) \subseteq \{ \alpha_i,\ldots,\alpha_n \}$ for all $i$, and
\item $\Var(s_{i,j}) \subseteq \{ \alpha_{i+1},\ldots,\alpha_n \}$ for all $i,j$.
\end{enumerate}
Then the sequent
\[
H\ =\
\begin{array}{l}
\{ F_i\unsubst{\bar{x}}{\widebar{u_{i,j}}} \mid  1 \leq i \leq p, 1\leq j \leq m_i\},A_1 \impl \Land_{j=1}^{k_1} A_1\unsubst{\alpha_1}{s_{1,j}},\ldots, 
A_n \impl \Land_{j=1}^{k_n} A_n\unsubst{\alpha_n}{s_{n,j}}\\
\seq \{ F_i\unsubst{\bar{x}}{\widebar{u_{i,j}}} \mid  p < i \leq q, 1\leq j \leq m_i \} 
\end{array}
\]
is called an {\em extended Herbrand-sequent} of $\Gamma\seq\Delta$ if $H$ is
a tautology.
\end{definition}
The notion of schematic extended Herbrand-sequent is defined analogously
to Definition~\ref{def.sch_ext_Hseq} by replacing
the formulas $A_i$ in the above definition by monadic predicate variables $X_i$.
The size of a (schematic) extended Herbrand sequent $H$ of the above form
is $|H|=\sum_{i=1}^q m_i + \sum_{j=1}^n k_j$.
\begin{theorem}
$\Gamma\seq\Delta$ has a proof with $\Pi_1$-cuts and $\|\pi\|_\mathrm{q} = l$ iff
it has an extended Herbrand-sequent $H$ with $|H|=l$.
\end{theorem}
\begin{proof}
Analogous to the proof of Theorem~\ref{thm.proof_extHseq},
replacing $F\unsubst{x}{U_i}$ by the collection of all
instances $F_i\unsubst{\bar{x}}{\widebar{u_{i,j}}}$ s.t.\ all terms in
$u_{i,j}$ contain only variables from $\{ \alpha_{i+1},\ldots,\alpha_n \}$.
\end{proof}
The above theorem encapsulates an algorithm for the construction of a proof
with $\Pi_1$-cuts from an extended Herbrand-sequent. We will henceforth
use the abbreviation {\PCA} for this proof-construction algorithm.

In order to represent term vectors, it is helpful to enrich our signature by
new function symbols $f_1,\ldots,f_q$ where $f_i$ has arity $l_i$. The function
symbol $f_i$ will serve the purpose of grouping a term-tuple which corresponds
to an instantiation of the formula $\forall x_1\cdots \forall x_{l_i} F_i$
if $i\leq p$ (or $\exists x_1\cdots \exists x_{l_i} F_i$ if $i > p$).
\begin{definition}\label{def.more_general_gram_extHseq}
The grammar of an extended Herbrand-sequent
\[
H\ =\
\begin{array}{l}
\{ F_i\unsubst{\bar{x}}{\widebar{u_{i,j}}} \mid  1 \leq i \leq p, 1\leq j \leq m_i\},A_1 \impl \Land_{j=1}^{k_1} A_1\unsubst{\alpha_1}{s_{1,j}},\ldots, 
A_n \impl \Land_{j=1}^{k_n} A_n\unsubst{\alpha_n}{s_{n,j}}\\
\seq \{ F_i\unsubst{\bar{x}}{\widebar{u_{i,j}}} \mid  p < i \leq q, 1\leq j \leq m_i \} 
\end{array}
\]
is defined as $G(H) = \langle N_R, \Sigma, \tau, P \rangle$ where $N_R = \{ \tau, \alpha_1,\ldots,\alpha_n \}$,
$\Sigma$ is the signature of $H$ plus $\{ f_1,\ldots, f_q\}$ and
$P = \{ \tau \rightarrow f_i(\widebar{u_{i,j}}) \mid 1\leq i \leq q, 1\leq j \leq m_i\}
\union \{ \alpha_i \rightarrow s_{i,j} \mid 1\leq i \leq n, 1\leq j \leq k_i \}$.
\end{definition}
Note that this definition also applies to Herbrand-sequents (as in Definition~\ref{def.herbrand.sequent}): then $n=0$ and we obtain a trivial grammar $\langle N_R, \Sigma, \tau, P \rangle$ 
with $N_R=\{\tau\}$. Using this grammar, we define the {\em Herbrand terms} of a 
Herbrand-sequent as the set $\{t \mid (\tau \rightarrow t)\in P\}$. The {\em Herbrand
terms of a cut-free proof $\pi$} are then the Herbrand terms of the Herbrand-sequent
extracted from $\pi$ via Theorem~\ref{thm.proof_Hseq2}.
\begin{example}
Consider the sequent
\[
P(0,0), \forall x\forall y\, (P(x,y) \impl P(s(x),y)),
\forall x\forall y\, (P(x,y) \impl P(x,s(y)) \seq P(s^4(0), s^4(0)).
\]
Abbreviating
$P(x,x)\impl P(s^2(x),s^2(x))$ as $F(x)$ we see that $\forall x\, F(x)$ is a
useful cut formula that allows to decrease the number of $\forall_\mathrm{l}^*$-inferences.
A corresponding extended Herbrand-sequent is
\begin{align*}
E = & P(0,0), P(\alpha,\alpha) \impl P(s(\alpha),\alpha), P(s(\alpha),\alpha) \impl P(s(\alpha),s(\alpha)),
P(s(\alpha),s(\alpha)) \impl P(s^2(\alpha),s(\alpha)),\\
 & P(s^2(\alpha),\alpha)\impl P(s^2(\alpha), s^2(\alpha)),
F(\alpha)\impl (F(0) \land F(s^2(0))) \seq P(s^4(0),s^4(0))
\end{align*}
The corresponding grammar is $\Gram(E) = \langle \{ \tau,\alpha \}, \{ 0/0, s/1
\}, \tau, P\rangle$, with $|E|=|G(E)|=6$ and
\[
P = \{ \tau \rightarrow f_1(\alpha,\alpha) | f_2(s(\alpha),\alpha) | f_1(s(\alpha),s(\alpha)) | f_2(s^2(\alpha), \alpha), \alpha \rightarrow 0 | s^2(0) \}.
\]
\end{example}
In extension of our compact notation for grammars we write
\[
(U_1,\ldots,U_q) \circ_{\alpha_1} S_1 \cdots \circ_{\alpha_n} S_n
\]
for the grammar $\Gram(H)$ of Definition~\ref{def.more_general_gram_extHseq}
where $U_i = \{ \widebar{u_{i,j}} \mid 1 \leq j \leq m_i \}$
and $S_i = \{ s_{i,j} \mid 1 \leq j \leq k_i \}$.
As before, we leave out the $\alpha_i$ if they are
obvious from the context.
In this notation the function
symbols $f_i$ are implicitly specified by the position of $U_i$ in the vector
$(U_1,\ldots,U_q)$.
Consequently, each $t\in\Lang( (U_1,\ldots,U_q) \circ_{\alpha_1} S_1 \cdots \circ_{\alpha_n} S_n )$
has one of the $f_i$ as top-level symbol and these are the only occurrences
of $f_i$. For notational convenience and if $l_i = 1$ for all $i$ we sometimes
write $\Lang( (U_1,\ldots,U_q) \circ_{\alpha_1} S_1 \cdots \circ_{\alpha_n} S_n )$
as a vector of sets of terms in the form $(T_1,\ldots, T_q)$ where
$T_i = \{ t\in\Lang( (U_1,\ldots,U_q) \circ_{\alpha_1} S_1 \cdots \circ_{\alpha_n} S_n ) \mid t = f_i(\bar{s})\ \mbox{for some}\ \bar{s} \}$.
\begin{theorem}
If $H$ is an extended Herbrand-sequent of $\Gamma\seq\Delta$, then
\[
\{ F_i\unsubst{\bar{x}}{\bar{t}} \mid 1\leq i \leq p, f_i(\bar{t}) \in \Lang(\Gram(H)) \}
\seq \{ F_i\unsubst{\bar{x}}{\bar{t}} \mid p < i \leq q, f_i(\bar{t}) \in \Lang(\Gram(H)) \}
\]
is a Herbrand-sequent of $\Gamma\seq\Delta$.
\end{theorem}
\begin{proof}
As for Theorem~\ref{thm.grammar_cutelim}.
\end{proof}
\begin{definition}
Let $H$ be a schematic extended Herbrand sequent. Define
\[
C_1 = \Land_{i=1}^p\Land_{j=1}^{m_i} F_i\unsubst{\bar{x}}{\widebar{u_{i,j}}}
\land \Land_{i=p+1}^q \Land_{j=1}^{m_i} \neg F_i\unsubst{\bar{x}}{\widebar{u_{i,j}}}
\ \mbox{and}\ 
C_{i+1}\ = \Land_{j=1}^{k_i} C_i\unsubst{\alpha_i}{s_{i,j}}\ \mbox{for}\ i=1,\ldots,n.
\]
Then
\[
\sigma := \unsubst{X_i}{\lambda \alpha_i. C_i}_{i=1}^n
\]
is called {\em canonical substitution} of $H$.
\end{definition}
\begin{lemma}\label{lem.can_sol_gen}
Let $H$ be a schematic extended Herbrand-sequent and $\sigma$ be its canonical
substitution. Then $\sigma$ is a solution of $H$.
\end{lemma}
\begin{proof}
As for Lemma~\ref{lem.can_sol}.
\end{proof}
As in Section~\ref{sec.pth_instruct}, the canonical substitution is hence called {\em canonical solution}.
\begin{theorem}
$\Gamma\seq\Delta$ has an extended Herbrand-sequent $H$ with $|H|=l$ iff there
is a totally rigid acyclic tree grammar $G$ with $|G|=l$ s.t.
\[
\{ F_i\unsubst{\bar{x}}{\bar{t}} \mid 1\leq i \leq p, f_i(\bar{t}) \in \Lang(\Gram(H)) \}
\seq \{ F_i\unsubst{\bar{x}}{\bar{t}} \mid p < i \leq q, f_i(\bar{t}) \in \Lang(\Gram(H)) \}
\]
is a tautology.
\end{theorem}
\begin{proof}
Analogous to the proof of Theorem~\ref{thm.extHseq_grammar}, using the canonical
solution to obtain an extended Herbrand-sequent from a grammar.
\end{proof}

\section{Efficient Computation of a Grammar}
\label{sec:computation_of_grammar}

Let $\pi$ be a cut-free proof and $T$ the set of (tuples of) terms used in rules
$\forall_l^*$ and $\exists_r^*$ in $\pi$.
From Theorem \ref{thm.proof_Hseq2} we conclude that $|T| = \|\pi\|_q$ and that $T$ is easily obtained
from the Herbrand sequent.

In this section we address the problem of obtaining a grammar $G$ such that
$L(G)=T$. As it was shown before in Theorems \ref{thm.proof_extHseq}
and \ref{thm.extHseq_grammar}, the size of $G$ will determine the quantifier complexity of the
proof with cuts. Since we are interested in reducing this complexity, the main
goal is to find a \emph{minimal grammar} $G$, such that $|G| < |T|$. Whenever this is
possible, we are able to construct a proof with cuts $\psi$ such that $\|\psi\|_q
< \|\pi\|_q$. Note that this might not be possible for every set of terms $T$. 

Given our representation for $G = U \circ S_1 \circ ... \circ S_n$, the size
$|G|$ is the same as $|U| + |S_1| + ... + |S_n|$. Since there is a bound on the
size of $G$, namely, $|T|$, the most naive algorithm for finding grammars would
be \emph{guessing} which terms occur in $S_i$ or $U$ and \emph{checking} whether
this grammar generates $T$. We refer to this algorithm as {\GG} for theoretical
purposes. However, in this section we describe a more efficient approach for
computing grammars that can be used in practice.
%

Assume $\pi$ is a proof of an end-sequent of the form:

\[
\forall x_1 \cdots \forall x_{l_1}\, F_1, \ldots, \forall x_1\cdots\forall x_{l_p} F_p \seq
\exists x_1\cdots\exists x_{l_{p+1}}\, F_{p+1},\ldots, \exists x_1\cdots\exists x_{l_q} F_q
\]

Its Herbrand sequent, as described in Section \ref{sec:more_general}, is:

\[
\{ F_i\unsubst{\bar{x}}{\widebar{t_{i,j}}} \mid  1 \leq i \leq p, 1\leq j \leq n_i\} \seq
\{ F_i\unsubst{\bar{x}}{\widebar{t_{i,j}}} \mid  p < i \leq q, 1\leq j \leq n_i \} 
\]

The terms used to instantiate the formulas $F_i$ are easily obtained from this
sequent: they are the vectors $t_{i,j}$. Let

\[
T = \{ f_i(t_{i,j}) \mid 1 \leq i \leq q \}.
\]

with $f_i$ being fresh function symbols as in Definition
\ref{def.more_general_gram_extHseq}. These function symbols will facilitate the
computation of a grammar, avoiding the need to deal with tuples of terms instead
of terms.
Then the algorithm described in this section will compute a grammar, as in
Definition \ref{def.more_general_gram_extHseq}, for the set of terms $T$.

We will start by describing how to compute a grammar of the form $U
\circ S$, which contains the terms used in a proof with one cut. Later, in
Section \ref{sec:generalization}, we show how to iterate this procedure in order
to get the grammar $U \circ S_1 \circ ... \circ S_n$ which will contain the
terms used in a proof with $n$ cuts.

The algorithm will rely on an operation called $\Delta$-vector, explained in
Section \ref{sec:deltavector}, and a data structure called $\Delta$-table,
explained in Section \ref{sec:deltatable}. Intuitively, the operation computes
``partial'' grammars that are stored in this table, which is later processed to obtain
grammars that generate the whole set $T$.

From now on we consider $T$ as a \emph{sequence} of terms instead of a set. This
will guide the search and help prune the search space.

\subsection{$\Delta$-vector}
\label{sec:deltavector}

The $\Delta$-vector of a sequence of terms $T$ describes the differences
between the terms in $T$. It is defined as:

%

\[
\Delta(t_1, ..., t_n) = \left\{
  \begin{array}{ll}
    (f(u_1, ..., u_m), (s_1, ..., s_n)) & \text{if all $t_i = f(t^i_1, ..., t^i_m)$ and}\\
    & \Delta(t^1_j,..., t^n_j) = (u_j, (s_1, ..., s_n))\;\;\forall\ j \in \{1,...,m\} \\
    (\alpha, (t_1, ..., t_n)) & \text{otherwise}\\
  \end{array} \right.
\]

where $\alpha$ is an eigenvariable.

For example, if $T = (fa, fb)$, its $\Delta$-vector is $(f\alpha, (a, b))$.
But in order to make this definition clearer, we will analyse a more involved example. Let
$T = \{ f(gc, c), f(g^2c, gc), f(g^3c, g^2c) \}$. Then:
\begin{equation}
\label{eq:delta1}
\Delta(f(gc, c), f(g^2c, gc), f(g^3c, g^2c)) = (f(u_1, u_2), (s_1, s_2, s_3))
\end{equation}
\begin{center}if\end{center}
\begin{equation}
\label{eq:delta2}
\Delta(gc, g^2c, g^3c) = (u_1, (s_1, s_2, s_3))
\end{equation}
\begin{equation}
\label{eq:delta3}
\Delta(c, gc, g^2c) = (u_2, (s_1, s_2, s_3))
\end{equation}
Note that the second element of the pair, the vector $(s_1, s_2, s_3)$ must be the same 
for the $\Delta$-vector of the arguments.

In order to solve Equation \ref{eq:delta2}, we apply the same definition:
\begin{equation}
\label{eq:delta4}
\Delta(gc, g^2c, g^3c) = (gu'_1, (s'_1, s'_2, s'_3))
\end{equation}
\begin{center}if\end{center}
\begin{equation}
\label{eq:delta5}
\Delta(c, gc, g^2c) = (u'_1, (s'_1, s'_2, s'_3))
\end{equation}
Since the terms in $(c, gc, g^2c)$ do not have a common head symbol, it's $\Delta$-vector
is: $(\alpha, (c, gc, g^2c))$. This solves Equations \ref{eq:delta3} and 
\ref{eq:delta5}, and $u'_1 = u_2 = \alpha$ and $(s'_1, s'_2, s'_3) = (s_1, s_2, s_3) = (c, gc, g^2c)$.
So now Equation \ref{eq:delta2} (and \ref{eq:delta4}) is:
\[
\Delta(gc, g^2c, g^3c) = (g\alpha, (c, gc, g^2c))
\]
And Equation \ref{eq:delta1} is solved:
\[
\Delta(f(gc, c), f(g^2c, gc), f(g^3c, g^2c)) = (f(g\alpha, \alpha), (c, gc, g^2c))
\]
%

It is worth to note that the $\Delta$-vector is
already a grammar $U \comp{\alpha} S$, but with the particularity of having only one term in
the set $U$ (represented by $f(u_1, ..., u_m)$). Thus, $P = \{ \tau \rightarrow
f(u_1, ..., u_m) \} \union \{ \alpha \rightarrow s_i | s_i \in S \}$.

\begin{definition}
\label{def:simplegrammars}
A grammar $U \comp{\alpha} S$ is called \emph{simple} if the set $U$ contains
only one term and it is called \emph{trivial} if it is simple and $U =
\{\alpha\}$, i.e., $\tau \rightarrow \alpha$ is the only derivation from the start
symbol.
\end{definition}

Observe that the $\Delta$-vector computes only \emph{simple} grammars.
In the following sections we show how to
combine these grammars to obtain more complex ones (Section
\ref{sec:deltatable}) and how to find valid grammars that will generate all the
terms from $T$ (Section \ref{sec:finddecomposition}).

\subsection{$\Delta$-table}
\label{sec:deltatable}

The $\Delta$-table is a data-structure that stores the \emph{non-trivial}
simple grammars $U \comp{\alpha} S$ computed by applying the $\Delta$-vector
exhaustively to sub-sequences of $T$.

The motivation behind the exhaustive procedure of computing the $\Delta$-vector
of all possible sub-sequences is to obtain the best possible compression
for the final grammar, i.e., $U \comp{\alpha} S$ such that $|U| + |S|$ is the least
possible (and, of course, less than $|T|$). Let us illustrate the
situation.

Given a sequence of terms $T = ( t_1, ..., t_n )$, the $\Delta$-vector of this
sequence is:

$$\Delta(t_1, ..., t_n) = (u_{\alpha}, (s_1, ..., s_n))$$

in which $u_{\alpha}$ is the biggest common term of all $t_i$ parametrized
with some variable $\alpha$ such that, replacing this variable with each $s_i$
would yield the original set $T$. As we said before, this is a grammar $G$ with
$P = \{ \tau \rightarrow u_{\alpha} \} \union \{ \alpha \rightarrow s_i \mid 1 \leq i \leq n
\}$, such that $L(G) = T$, but it's not a good one. 
Observe that $U$ would have only one term
($u_{\alpha}$) and $S$ has the same number of terms as the input, i.e., $n$. So
$|G| = n+1$, which is bigger than $|T|$. Since we are
interested in finding grammars that compress the size of the term sequence,
this is not a good choice.

By computing the $\Delta$-vector of sub-sequences of $T$ and combining them, we can
obtain such a compression.  To make this clearer, consider the term sequence:
$$
T = ( f(c, gc), f(c, g^2c), f(c, g^3c), f(gc, c), f(g^2c, gc), f(g^3c, g^2c) )
$$
The $\Delta$-vector of this sequence is the following trivial grammar:
$$
\alpha \comp{\alpha} ( f(c, gc), f(c, g^2c), f(c, g^3c), f(gc, c), f(g^2c, gc),
f(g^3c, g^2c) )
$$
But if we take well-chosen subsets of this set, we obtain the pairs:
\begin{align*}
\Delta(f(c, gc), f(c, g^2c), f(c, g^3c)) &= ( f(c, g\alpha), (c, gc, g^2c ) )\\
\Delta(f(gc, c), f(g^2c, gc), f(g^3c, g^2c)) &= ( f(g\alpha, \alpha), (c, gc, g^2c) )
\end{align*}
Now let $U = ( f(c, g\alpha), f(g\alpha, \alpha) )$ and $S = (c, gc,
g^2c )$. Note that $L(U \comp{\alpha} S) = T$, and $|U| + |S| = 2 + 3 = 5$. 
In particular, the combination of the first term of $U$ with the terms from $S$
generates the first 3 elements of $T$, and the combination of the second term of
$U$ with the terms from $S$ generates the last 3 elements from $T$.

The computation of these grammars is done incrementally, starting from
sub-sequences of size 1 until $n$, where $n$ is the size of the term sequence
$T$. The results are stored in a map\footnote{A map is a data structure that
stores values indexed by keys.}, called $\Delta$-table. The values stored in
this map are lists of pairs $(u, T)$, where $u$ is a term and $T$ is a set of
terms. They are indexed by another set of terms $S$.

For example, let $T' \subset T$ and $\Delta(T') = (u, (s_1, ..., s_k))$. This
information is stored in $T$'s $\Delta$-table with $S = (s_1, ..., s_k)$ as
the key and (a list of) $(u, T')$ as the value. Since there might be other
sub-sequences $T''$ of $T$ such that $\Delta(T'') = (u', S)$, it is necessary to store a
list of pairs.

Algorithm~\ref{alg:deltatable} creates, fills and returns the $\Delta$-table
for a sequence of terms $T$. It is important to note that this is different from
the naive algorithm, which would just compute and store the $\Delta$-vectors of
all sub-sequences of $T$. Algorithm~\ref{alg:deltatable} allows a significant
pruning of the search space which is based on the following theorem:

\begin{theorem}
Let $T$ be a set of terms. If $\Delta(T) = (\alpha, T)$ (trivial grammar),
then $\Delta(T') = (\alpha, T')$ for every $T' \supset T$.
\end{theorem}

\begin{proof}
Follows from the definition of $\Delta$-vector.
\end{proof}

Instead of computing the $\Delta$-vector for all sub-sequences of $T$, the algorithm searches the
$\Delta$-table for these sub-sequences\footnote{Note that the $\Delta$-table is
initialized with an empty grammar.} and tries to increase their size by one
element on each iteration. Since trivial grammars are not stored in the
$\Delta$-table, the algorithm never tries to increase the size of some set $T'$
when $\Delta(T') = (\alpha, T')$, thus avoiding large areas of the search space.

\begin{algorithm}
\caption{Fill $\Delta$-table for a sequence of terms $T$}
\label{alg:deltatable}
\begin{algorithmic}
\Function{$\Delta$-table}{$T$: sequence of terms}
  \State $table \gets$ \textbf{new} HashMap
  \State $table[[]] \gets [(\texttt{null}, [])]$
  \For{$i = 1 \to T.length$}
    \For{$(u, T') \in table | T'.length = i-1$}
      \For{$t \in T$ and $t \notin T'$}
        \Comment{$T[i] = t$ and $\forall t' \in K, T[j] = t', i > j$}
        \State $(u', S) \gets \Delta$-vector$(T' + t)$
        \If{$u' \neq \alpha$}
          \State $table[S] \gets table[S] + (u', (T' + t))$
        \EndIf
      \EndFor
    \EndFor
  \EndFor
  \State \Return $table$
\EndFunction
\end{algorithmic}
\end{algorithm}

\subsection{Finding valid grammars}
\label{sec:finddecomposition}

After having filled the $\Delta$-table, it is only a matter of combining the
simple grammars in order to find a suitable one, i.e., a grammar $G$ such that
$L(G) = T$.

Let $S \rightarrow [(u_1, T_1), ..., (u_r, T_r)]$ be one entry of $T$'s
$\Delta$-table. We know that $T_i \subset T$ and
that $(u_i, S)$ is a grammar $G_i$ such that $L(G_i) = T_i$ for each $i \in \{1 ... r\}$. 
Take $\{T_{i_1}, ..., T_{i_s}\} \subset \{T_1, ..., T_r\}$ such that $T_{i_1} \union
... \union T_{i_s} = T$. Then, since combining each $u_{i_j}$ with $S$
yields $T_{i_j}$, and the union of these terms is $T$, the grammar
$(u_{i_1}, ..., u_{i_s}) \comp{\alpha} S$ will generate all terms from $T$. Note that there might
be several combinations of $T_i$ such that its union covers the set $T$, so it
is often the case that different grammars are found, but only the minimal ones
are considered as possible solutions.

It might happen that there are no $\{T_{i_1}, ..., T_{i_s}\} \subset \{T_1, ...,
T_r\}$ such that $T_{i_1} \union ... \union T_{i_s} = T$, but still a
compression of the terms is possible. This is the case, for example, of the set
$T = \{a, fa, f^2a, f^3a \}$. The $\Delta$-table built for this set of terms is
the following:

\begin{align*}
\{ a, fa \} &\Rightarrow [(f\alpha, \{ fa, f^2a\}), (f^2a, \{ f^2a, f^3a \})]\\
\{ a, f^2a \} &\Rightarrow [(f\alpha, \{ fa, f^3a \})]\\
\{ a, fa, f^2a \} &\Rightarrow [(f\alpha, \{ fa, f^2a, f^3a\})]\\
\end{align*}

Nevertheless, there are grammars other than the trivial one that generate this
term set, e.g., $\{\alpha, f\alpha \} \circ \{ a, f^2a \}$. In order to find these grammars,
some trivial grammars must be added to the $\Delta$-table. Remember that
these were removed to reduce the search space while building the table, but it is still desirable that
every possible grammar is found. This problem is solved during the search for a
valid grammar, after the $\Delta$-table is completed.
Observe that it only makes sense to add a trivial grammar $\{ \alpha \} \circ
T_i$ if $T_i \subset T$. Therefore, for every entry of the $\Delta$-table such that
the key is a set $T_i \subset T$, the trivial grammar is added. Thus the new
$\Delta$-table of the example would be:

\begin{align*}
\{ a, fa \} &\Rightarrow [(f\alpha, \{ fa, f^2a\}), (f^2a, \{ f^2a, f^3a \}), (\alpha, \{ a, fa\})]\\
\{ a, f^2a \} &\Rightarrow [(f\alpha, \{ fa, f^3a \}), (\alpha, \{ a, f^2a \})]\\
\{ a, fa, f^2a \} &\Rightarrow [(f\alpha, \{ fa, f^2a, f^3a\}), (\alpha, \{ a, fa, f^2a \})]\\
\end{align*}

And from this, the grammar $\{\alpha, f\alpha \} \circ \{ a, f^2a \}$ can be obtained.

\subsection{Generalization to multiple cuts}
\label{sec:generalization}

In the previous sections it was explained how to compute a grammar $U
\comp{\alpha} S$ for a set of terms $T$.
If the term set was extracted from a proof of a skolemized end-sequent, then
this allows the introduction of one $\Pi_1$-cut ($\forall x. C$) in a proof of the same
end-sequent. In this new proof, the end-sequent formulas will be instantiated
with the terms from the set $U$ of the grammar (these terms are prefixed with
the function symbol $f_i$, which indicates of which formula these terms should
be instances). The terms from the set $S$ will be used to instantiate the cut
formula when it occurs on the left, and $\alpha$ will be used as the
eigenvariable of the cut formula on the right.

In order to obtain a grammar for constructing a proof with multiple cuts, all
that is needed is to iterate the procedure described in the previous sections.
Remember that, a totally rigid acyclic grammar for a proof with
$n$ cuts can be represented by\footnote{Note that this representation has the
indices reversed from the usual representation used so far. This is only
because, in practice, the number $n$ in which the iteration stops is not known
in advance.}:

$$U \comp{\alpha_n} S_n \dots \comp{\alpha_1} S_1$$

Given a sequence of terms $T$, on the first step the algorithm will compute a
grammar $U_1 \comp{\alpha_1} S_1$ such that it generates (and compresses) $T$.
Then, the algorithm is run again, now with input $U_1$ (and with $\alpha_1$
considered as a constant), in an attempt to
compress even more the term set. Suppose that a grammar $U_2 \comp{\alpha_2}
S_2$ was found for $U_1$, such that it still compresses this term set. Now we
have the grammar $U_2 \comp{\alpha_2} S_2 \comp{\alpha_1} S_1$ that generates the
original term set $T$. This procedure can continue until we reach a term set
$U_n$ that cannot be compressed anymore via a grammar. At this moment we stop
computing grammars, and we can compose them all to generate $T$. This procedure
is illustrated in Figure \ref{fig:iterateDec}.

\begin{figure}[h]
\begin{align*}
T &= L(U_1 \circ_{\alpha_1} S_1)\\
U_1 &= L(U_2 \circ_{\alpha_2} S_2)\\
&\vdots\\
U_{n-1} &= L(U_n \circ_{\alpha_n} S_n)
\end{align*}
\caption{For more than one cut, iterate the grammars.}
\label{fig:iterateDec}
\end{figure}


The algorithm to compute grammars described in the previous sections will be
henceforth referred to as \DA.

\subsection{Example}

In this section we show the computation of the grammars for an actual cut-free
proof. For readability reasons we do not show the computation of every
$\Delta$-vector nor the full $\Delta$-table, but only those relevant to find a
minimal grammar.


The following sequent:

$$P(0, 0), \forall x  \forall y (P(x, y) \supset P(x, sy)), \forall x  \forall y (P(x, y) \supset P(sx, y)) \seq P(s^40, s^40) $$

has a cut-free proof whose Herbrand-sequent is:

\[
\begin{array}{lcl}
\begin{array}{l}
P(0,0),\\
(P(0, 0) \supset P(s0, 0)),\\
(P(s0, 0) \supset P(s0, s0)),\\
(P(s0, s0) \supset P(s^20, s0)),\\
(P(s^20, s0) \supset P(s^20, s^20)),\\
(P(s^20, s^20) \supset P(s^30, s^20)),\\
(P(s^30, s^20) \supset P(s^30, s^30)),\\
(P(s^30, s^30) \supset P(s^40, s^30)),\\
(P(s^40, s^30) \supset P(s^40, s^40))
\end{array}
&
\seq
&
P(s^40, s^40)
\end{array}
\]

Let $F_1 = \forall x \forall y (P(x, y) \supset P(x, sy))$ and $F_2 = \forall x  \forall
y  (P(x, y) \supset P(sx, y)) $. Then, the tuple term sets $T_1$ and $T_2$ of
respectively $F_1$ and $F_2$ are extracted:

\[
\begin{array}{rcl}
T_1 &=& ((s0, 0), (s^20, s0), (s^30, s^20), (s^40, s^30)) \\
T_2 &=& ((0, 0), (s0, s0), (s^20, s^20), (s^30, s^30))
\end{array}
\]

As in Definition~\ref{def.more_general_gram_extHseq} we will use two fresh
function symbols $f_1$ and $f_2$, both of arity 2, to build one set of terms
from the tuples of terms used to instantiate the formulas $F_1$ and $F_2$
respectively:

\[
\begin{array}{ll}
T = \{ &f_1(s0, 0), f_1(s^20, s0), f_1(s^30, s^20), f_1(s^40, s^30),\\
&f_2(0, 0), f_2(s0, s0), f_2(s^20, s^20), f_2(s^30, s^30)  \} 
\end{array}
\]

This term set has size 8, and our goal is to find a grammar $U \comp{\alpha} S$
such that $|U| + |S| < 8$. The first step of the algorithm is to fill the
$\Delta$-table by computing all non-trivial $\Delta$-vectors of subsets
of $T$. In particular, these two
$\Delta$-vectors are computed from the first four and last four elements of $T$:

\[
\begin{array}{rcl}
\Delta(f_1(s0, 0), f_1(s^20, s0), f_1(s^30, s^20), f_1(s^40,
s^30)) &=& (f_1(s\alpha, \alpha), ( s^30, s^20, s0, 0 )) \\
\Delta(f_2(0, 0), f_2(s0, s0), f_2(s^20, s^20), f_2(s^30, s^30)) &=&
(f_2(\alpha, \alpha), ( s^30, s^20, s0, 0 ) ) 
\end{array}
\]

These will be stored in the $\Delta$-table, which thus will have the following
entry:

\[
\begin{array}{ll}
\{ s^30, s^20, s0, 0 \} \Rightarrow &[ (f_1(s\alpha, \alpha), \{ f_1(s0, 0), f_1(s^20, s0), f_1(s^30, s^20), f_1(s^40,
s^30)\}),\\
&(f_2(\alpha, \alpha), \{ f_2(0, 0), f_2(s0, s0), f_2(s^20, s^20), f_2(s^30, s^30)\} ) ]\\
\end{array}
\]

Given these entries, the algorithm finds the following grammar that generates
$T$:

\[
\begin{array}{rcl}
\{ f_1(s\alpha, \alpha), f_2(\alpha, \alpha) \} &\circ& \{ s^30, s^20, s0, 0 \}\\
\end{array}
\]

In fact, for this example, the algorithm finds 31 grammars, of which 3 have the
minimal size 6. From Theorems \ref{thm.proof_extHseq} and
\ref{thm.extHseq_grammar}, we know that grammars of size $l$
generate proofs $\pi$ with $\Pi_1$-cuts such that $|\pi|_q = l$. Therefore,
initially, all minimal grammars are equally good. But in Section
\ref{sec.implementation_experiments} we mention another heuristic to decide
which grammar is used.

\section{Improving the canonical solution}
\label{sec.improving}
After completing the first phase of cut-introduction, namely the computation
of a grammar,
the next step is to find a solution to the schematic extended Herbrand
sequent induced by the grammar. Such a solution is guaranteed to exist by
Lemma~\ref{lem.can_sol}, and its construction is described in
Definition~\ref{def.can_subst}. But is this solution optimal? If we approach
this question from the point of view of the $|\cdot|_\mathrm{q}$ measure, 
Theorem~\ref{thm.proof_extHseq} shows that all solutions can be considered
equivalent. From the point of view of symbolic complexity or logical
complexity, things may be different: there are cases where the canonical solution
is large, but small solutions exist. The following example exhibits
such a case. In this example, a smaller solution not only exists, but is also
more natural than (and hence in many applications preferable to) the canonical solution.
\begin{example}\label{ex.improving}
Consider the sequents
\[
S_n\qequiv Pa, \forall x\, (Px\impl Pfx) \seq Pf^{n^2}a.
\]
Note that this is the example from Section~\ref{sec.motex} where $2^n$ is replaced
by $n^2$. $S_n$ has a (minimal) Herbrand-sequent
\[
H_n\qequiv Pa, Pa\impl Pfa,\ldots,Pf^{n^2-1}a \impl Pf^{n^2}a \seq Pf^{n^2}a.
\]
The terms of this Herbrand-sequent are generated by the grammar
\[
\{\alpha, f\alpha,\ldots,f^{n-1}\alpha\} \circ \{a,f^na,\ldots,f^{(n-1)n}a\}
\]
which gives rise to the schematic extended Herbrand-sequent
\[
  X(\alpha) \impl \bAND_{i=0}^{n-1} X(f^{in}a), Pa, P\alpha \impl Pf\alpha,\ldots,Pf^{n-1}\alpha \impl Pf^n\alpha \seq Pf^{n^2}a
\]
and the canonical solution $\sigma=\unsubst{X}{\lambda \alpha . C}$
with
\[
  C\qequiv Pa\land\bAND_{i=0}^{n-1}(Pf^i\alpha\impl Pf^{i+1}\alpha)\land\neg Pf^{n^2}a.
\]
But there also exists a solution $\theta$ of constant logical complexity and linear
(instead of quadratic) symbol complexity
by taking $\theta=\unsubst{X}{\lambda \alpha . A}$ with
\[
A\qequiv P\alpha \impl Pf^n\alpha.
\]
\end{example}
Since the solution for the schematic extended Herbrand sequent is interpreted as the
lemmata that give rise to the proof with cuts, and these lemmata will in applications
be read and interpreted by humans, it is important to consider the problem of improving
the logical and symbolic complexity of the canonical solution. Furthermore,
a decrease in the logical complexity of a lemma often yields a decrease
in the length of the proof that is constructed from it. 

In the following sections, we will describe a method which computes small solutions for
schematic Herbrand sequents induced by grammars. The method will be abstract; it will
depend on an algorithm $\Cons$ enumerating consequences of a formula. We describe two concrete 
consequence generators: one
will be complete (but expensive), the other will be incomplete but less expensive.

We start by investigating
the case of a single $\Pi_1$-cut in the subsequent Section~\ref{sec.improving.single} 
(some of these results have essentially been presented already in~\cite{Hetzl12Towards}). 
We then describe the two consequence generators. Finally, we present an approach to the simplification of the canonical
solution for an arbitrary number of $\Pi_1$-cuts in Section~\ref{sec.improving.multiple}.

For simplicity of presentation we will consider a fixed sequent
\[
  S\qequiv \forall x\,F(x)\seq
\]
although the results can be extended to more general end-sequents
as in Section~\ref{sec:more_general}. The problem of improving the canonical solution is a
propositional one, hence in the sequel, $\alpha$ is to be interpreted as a constant symbol,
$\models$ denotes the propositional consequence relation, and all formulas are quantifier-free
unless otherwise noted.
\subsection{Improving the solution of a single $\Pi_1$-cut}\label{sec.improving.single}
We start the study of the problem of the simplification of the canonical solution
by looking at the case of {\em 1-grammars} $U\circ V$, which give rise to proofs with a single $\Pi_1$-cut.
In the setting of 1-grammars, a solution is of the form $\unsubst{X}{\lambda \alpha. A}$.
Throughout this section, we consider a fixed 1-grammar $U\circ V$, along with a schematic
Herbrand sequent 
\[
  H\qequiv F\unsubst{x}{u_1},\ldots,F\unsubst{x}{u_m},X(\alpha) \impl \Land_{j=1}^{k} X(s_{j})\seq
\]
and its canonical solution
\[
\sigma=\unsubst{X}{\lambda \alpha .C}=\unsubst{X}{\lambda \alpha .\bAND_{i=1}^mF\unsubst{x}{u_i}}.
\]
We will use the abbreviation 
\[
  \Gamma=F\unsubst{x}{u_1},\ldots,F\unsubst{x}{u_m}.
\]
If $\unsubst{X}{\lambda \alpha. A}$
is a solution for $H$, we will say simply that $A$ is a solution.

The first basic observation is that solvability is a semantic property. The following is
an immediate consequence of Definition~\ref{def.sch_ext_Hseq}. 
\begin{lemma}\label{lem:sol_equiv}
Let $A$ be a solution, $B$ a formula and $\models A \Leftrightarrow B$.
Then $B$ is a solution.
\end{lemma}
Hence we may restrict our attention to solutions which are
in {\em conjunctive normal form} (CNF). Formulas in CNF can be represented 
as sets of {\em clauses}, which in turn are sets of {\em literals}, i.e.~possibly
negated atoms. It is this representation that we will use throughout this
section, along with the following properties: for sets of clauses $A, B$,
$A\subseteq B$ implies $B\models A$, and for clauses $C,D$, 
$C\subseteq D$ implies $C\models D$.

Note that the converse of the Lemma above does not hold:
given a solution $A$ there may be
solutions $B$ such that $\nmodels A\Leftrightarrow B$. We now
turn to the problem of finding such solutions.
%
%
In Example~\ref{ex.improving}, we observe that 
that $C\models A$ (but $A \not\models C$).
We can generalize this observation to show that the canonical solution is most general.
\begin{lemma}\label{lem:can_gen}
Let $C$ be the canonical solution and $A$ an arbitrary solution. Then $C\models A$.
\end{lemma}
\begin{proof}
  Since $\vartheta=\unsubst{X}{\lambda \alpha.A}$ is a solution for $H$,
$H\vartheta=F\unsubst{x}{u_1},\ldots,F\unsubst{x}{u_m},A \impl \Land_{j=1}^{k} A\unsubst{\alpha}{s_{j}}\seq$ is valid.
By definition, $C=\bAND_{i=1}^{m} F\unsubst{x}{u_i}$, and therefore
$C, A\impl \Land_{j=1}^k A\unsubst{\alpha}{s_j} \seq$ is valid, hence $C\seq A$ is valid.
\end{proof}
This result states that any search for simple solutions can be restricted to consequences
of the canonical solution. Theoretically, we could simply enumerate ``all'' such consequences
(there are, up to logical equivalence, only finitely many), but of course this is computationally
infeasible. Towards a more efficient (but still complete!) iterative solution,
we give a criterion that allows us to disregard some of those consequences.
\begin{lemma}\label{lem:res_cip}
If $A\models B$ then
\begin{enumerate}
\item[(1)] 
If $A\unsubst{\alpha}{s_1},\ldots,A\unsubst{\alpha}{s_k},\Gamma\seq$ is not valid, then
$B$ is not a solution.
\item[(2)] If $A$ is a solution then $\Gamma\seq B$ is valid.
\item[(3)] If $A$ is a solution, then $B\unsubst{\alpha}{s_1},\ldots,B\unsubst{\alpha}{s_k},\Gamma\seq$
is valid iff $\unsubst{X}{\lambda \alpha . B}$ is a solution of $H$.
\end{enumerate}
\end{lemma}
\begin{proof}
For (1), we will show the contrapositive.
By assumption, $B\unsubst{\alpha}{s_1},\ldots,B\unsubst{\alpha}{s_k},\Gamma\seq$
is valid. Since $A \models B$, we find that $A\unsubst{\alpha}{s_1},\ldots,A\unsubst{\alpha}{s_k},\Gamma\seq$
is valid.
For (2) it suffices to observe that since $A$ is a solution $\Gamma\seq A$ is valid,
and to conclude by $A\models B$.
(3) is then immediate by definition.
\end{proof}
\begin{lemma}[Sandwich Lemma]\label{lem:sandwich}
Let $A, B$ be solutions and $A\models D\models B$. Then $D$ is a solution.
\end{lemma}
\begin{proof}
By Lemma~\ref{lem:res_cip}~(2), $\Gamma\seq D$ is valid.
By Lemma~\ref{lem:res_cip}~(1), $D\unsubst{\alpha}{s_1},\ldots,D\unsubst{\alpha}{s_k},\Gamma\seq$
is valid.
\end{proof}

Another observation to be made in Example~\ref{ex.improving} is that
$A$ only contains clauses that
contain $\alpha$. This observation can be generalized as well:
we may freely delete $\alpha$-free clauses from solutions.
%
\begin{lemma}\label{lem:clean_ground}
Let $A$ be a solution in CNF and 
$A'$ be obtained from $A$ by removing all clauses that
do not contain $\alpha$. Then $A'$ is a solution.
\end{lemma}
\begin{proof}
In this proof, we will denote $C\unsubst{\alpha}{t}$ by $C(t)$ for clauses $C$
and terms $t$. Note that since $A$ is a solution, the sequent
$s: \Gamma\seq A$ is 
valid.
Furthermore, note that the validity of
\[
  h: \Gamma,A' \impl \Land_{j=1}^{k} A'(s_{j})\seq
\]
follows from the validity of $h_1:A'(s_1),\ldots,A'(s_k),A\seq$
and $h_2:A \seq A'$ since $\Gamma\seq A'$ can be derived
from $s$ and $h_2$, and $\Gamma,A'(s_1),\ldots,A'(s_k)\seq$
can be derived from $s$ and $h_1$.

For $h_2$, validity is clear since
$A'\subseteq A$. 
Now let $A=\bAND_{i=1}^{l-1}C_i\land\bAND_{i=l}^{m}C_{i}$ 
such that $C_i$ contains $\alpha$ if and only if $l\leq i\leq m$.
Then $A'=\bAND_{i=l}^{m}C_{i}$.
By Lemma~\ref{lem.can_valid} and the application of the invertible
$\andl$ rule, we know that
\[
  \bAND_{i=1}^{l-1}C_i(s_1),\bAND_{i=l}^mC_i(s_1),\ldots,\bAND_{i=1}^{l-1}C_i(s_k),\bAND_{i=l}^mC_i(s_k) \seq
\]
is valid. By assumption, $C_i(t)=C_i$ for all terms $t$ and $1\leq i <l$, so by contraction we derive
\[
\bAND_{i=1}^{l-1}C_i,\bAND_{i=l}^mC_i(s_1),\ldots,\bAND_{i=l}^mC_i(s_k) \seq
\]
which implies $h_1$ since $A'(s_j)=\bAND_{i=l}^mC_i(s_j)$ and $A \seq C_i$ is valid for all $i,j$.
\end{proof}
We have now established all the results required for our method. Before we describe it, we need
some more notions.
A solution $A$ in CNF is called a {\em minimal solution}
if it has minimal symbol complexity among all solutions in CNF.
\begin{definition}
  Let $\Cons$ be an algorithm such that 
  $\Cons(A)$ is a finite set of propositional
  consequences of a formula $A$ (a {\em consequence generator}).
  We say that {\em $\Cons$ generates $A$ for $B$} if either
\begin{enumerate}
  \item $A\in\Cons(B)$ or
  \item there exists $B'\in\Cons(B)$ such that $\Cons$ generates
    $A$ for $B'$.
\end{enumerate}
$\Cons$ is {\em complete w.r.t.~$A$} if, for all minimal solutions $B$, $A\models B$ implies
that $\Cons$ generates $B$ for $A$.
$\Cons$ is {\em well-founded} 
if there exists a well-founded order $>$ such that $\Cons(A)>\Cons(A')$
for $A'\in\Cons(A)$. 
\end{definition}
Let $\Cons$ be a consequence generator and $A$ a solution.
Algorithm~\ref{alg:simp_sol} then describes
the solution-finding algorithm $\SF{\Cons}$.
\begin{algorithm}
\caption{\SF{\Cons}}
\label{alg:simp_sol}
\begin{algorithmic}
\Function{\SF{\Cons}}{$A$: solution in CNF}
  \State $A \gets A$ without $\alpha$-free clauses 
  \State $S \gets \{A\}$
  \For{$B \in \Cons(A)$}
    \If{$B\unsubst{\alpha}{s_1},\ldots,B\unsubst{\alpha}{s_k},\Gamma\seq$ is valid}
      \Comment $B$ is a solution
      \State $S \gets S\cup \SF{\Cons}(B)$
    \EndIf
  \EndFor
  \State \Return $\min(S)$
\EndFunction
\end{algorithmic}
\end{algorithm}
\begin{theorem}\label{thm:solution_finding}
Let $\Cons$ be a consequence generator, let $C$ be
the canonical solution in CNF, and let
$F$ be a minimal solution in CNF.
If $\Cons$ is complete w.r.t.~$C$, then $F\in\SF{\Cons}(C)$. 
If $\Cons$ is well-founded, then $\SF{\Cons}$ terminates
on every input.
\end{theorem}
\begin{proof}
Termination is trivial. For completeness, note that
$\alpha$-free clauses can be removed from $C$ by Lemma~\ref{lem:clean_ground}.
Since $C$ is a solution, for $B\in\Cons(C)$ 
it suffices to check
by Lemma~\ref{lem:res_cip}~(3)
whether $B\unsubst{\alpha}{s_1},\ldots,B\unsubst{\alpha}{s_k},\Gamma\seq$ 
is valid to determine whether $B$ is a solution.
If it is not valid, then we know by Lemma~\ref{lem:res_cip}~(1)
that no iteration of $\Cons$ on $B$ will yield a solution. Since $\Cons$ is complete
w.r.t.~$C$, all minimal solutions will be generated.
\end{proof}
\subsection{Simplification by deductive closure}\label{ssec:saturation}
In this and the following section, we describe two concrete
consequence generators. Both will be based on the 
propositional resolution rule
which we now recall.
Let $A=\{C_i\mid 1\leq i \leq n\}$ be a formula in CNF with clauses $C_i=\{L_{i,j} \mid 1\leq j \leq n_i\}$, where
the $L_{i,j}$ are literals. 
By $\dual{L}$ we denote the dual of a literal $L$.
For two clauses $C_i, C_j$, 
if there exists exactly
one pair $(k, l)$ such that
$L_{i,k}=\dual{L_{j,l}}$, we define their {\em propositional resolvent} as the clause
\[
  \res(C_i, C_j)=(C_i\setminus L_{i,k})\cup(C_j \setminus L_{j,l})
\]
and leave $\res(C_i, C_j)$ undefined otherwise. We define the {\em deductive closure} 
$\dedClos(A)$ as the least superset of $A$ such that
for all $C_1,C_2\in\dedClos(A)$ there exists a $C\in\dedClos(A)$ such that $C\subseteq \res(C_1,C_2)$.
It is well-known that $\mathcal{D}(A)$ is finite and can be computed from $A$ by repeated application 
of $\res(\cdot,\cdot)$.
Finally, we define the {\em subset consequence generator}
\[
  \dedClosOp(A)=\{B\subset A \mid |B|=|A|-1\}.
\]
where $A,B$ are sets of clauses.
Towards showing completeness of $\dedClosOp$, 
we recall (a consequence of) a result from~\cite{Lee1967}:
\begin{theorem}\label{thm:lee}
Let $A$ be a formula in CNF and $C$ be a non-tautological clause such that $A \models C$. Then
there exists $C'\in\dedClos(A)$ such that $C'\subseteq C$.
\end{theorem}
\begin{theorem}
Let $C$ be the canonical solution. Then $\dedClosOp$ is well-founded and
complete w.r.t.~$\dedClos(C)$.
Hence $F\in\SF{\dedClosOp}(\dedClos(C))$ for all minimal solutions $F$, 
and $\SF{\dedClosOp}$ always terminates.
\end{theorem}
\begin{proof}
$\dedClosOp$ is trivially well-founded.
Now let $\dedClos(C)\models B$ with $B=\{ B_i \mid 1\leq i\leq n\}$ a minimal solution.
Since $\dedClos(C)$ is
logically equivalent to $C$, $C\models B$ and by Theorem~\ref{thm:lee}, 
there exists $B'=\{B_1',\ldots,B_n'\}\subseteq \dedClos(C)$ such that
$B_i'\subseteq B_i$ and hence $C\models B'\models B$. 
By Lemma~\ref{lem:sandwich}, 
$B'$ is a solution, and $B'=B$ by minimality. We conclude by observing
that $\dedClosOp$ generates all $A\subseteq \dedClos(C)$ for $\dedClos(C)$,
hence $B$ in particular.
\end{proof}
\subsection{Simplification by forgetful resolution}\label{ssec:forgetful}
This section proposes another particular consequence generator that will
yield a more practical but incomplete algorithm based on the {\em forgetful resolution operator} which
resolves two clauses and then ``forgets'' them. Letting $A=\{C_i\mid 1\leq i \leq n\}$ we define
\[
  \FResOp(A)=\{ \{ \res(C_i,C_j) \}\cup (A\setminus\{C_i,C_j\})\mid 1\leq i<j\leq n, \res(C_i,C_j)\textrm{ defined}\}.
\]
Resolution is sound, hence $\FResOp$ is a consequence generator yielding the algorithm $\SF{\FResOp}$.
Furthermore, note that if $A'\in\FResOp(A)$ then $\com{A}>\com{A'}$, hence $\FResOp$ is well-founded
and $\SF{\FResOp}$ terminates on all inputs by Theorem~\ref{thm:solution_finding}.

We conclude our investigation of the improvement of the canonical solution in the case
of 1-grammars by applying $\SF{\FResOp}$ to a concrete example.
\begin{example}
Consider the sequent $S_n$ from Example~\ref{ex.improving} for $n=2$.
The canonical solution of $S_2$, written in conjunctive normal form, is
\[
C\qequiv Pa \land (\neg P\alpha \lor Pf\alpha) \land (\neg Pf\alpha \lor Pf^2\alpha) \land \neg Pf^4a.
\]
Application of Lemma~\ref{lem:clean_ground} yields
\[
C'\qequiv (\neg P\alpha \lor Pf\alpha) \land (\neg Pf\alpha \lor Pf^2\alpha).
\]
We have $\FResOp(C')=\{\neg P\alpha \lor Pf^2\alpha\}$.
By (2) of Lemma~\ref{lem:res_cip}, it suffices to check whether
\[
Pa, \neg Pa \lor Pf^2a, \neg Pf^2a \lor Pf^4a \seq Pf^4a
\]
is valid, which is the case. Since $\FResOp(\neg P\alpha \lor Pf^2\alpha)=\emptyset$, search terminates
and $\SF{\FResOp}$ has found the smaller solution $\neg P\alpha \lor Pf^2\alpha$.
\end{example}
\subsection{Improving the solution of multiple $\Pi_1$-cuts}\label{sec.improving.multiple}
This section is concerned with finding small solutions in the setting of
grammars $U\circ S_1\circ\cdots\circ S_n$, i.e.~in the setting of introduction of
$n$ cuts. As in the previous section, the problem of finding a minimal solution
is trivially decidable --- our aim here is to find an algorithm that
traverses the search space in a manner that takes into account the simplifying
results we have established so far.

The algorithm we present will be incomplete independently of whether $\Cons$ is complete,
but it will avoid the problem of having to deal at once with all
the components of the canonical solution.
More precisely, our algorithm will be based on an iteration of the algorithm $\SF{\Cons}$
for the $1$-cut introduction problem presented in Section~\ref{sec.improving.single}, and
hence gives rise to two concrete algorithms by plugging in the consequence generators
of Sections~\ref{ssec:saturation}~and~\ref{ssec:forgetful}.

Before we start to describe the algorithm, we will make
a short detour to define Herbrand-sequents for (some) non-prenex sequents.
This is done since, even though we have fixed a particularly simple sequent $S$,
such more general sequents will naturally appear in the description of the algorithm. 
For the remainder of this section, for notational simplicity
we will write $F_i(t)$ for $F_i\unsubst{\alpha_i}{t}$.
So consider sequents of the form
\[
M\qequiv F_1(\alpha_1) \impl \forall x_1 F_1(x_1),
\ldots, F_n(\alpha_n) \impl \forall x_n F_n(x_n),
\forall x F(x) \seq 
\]
such that $x_i\notin \Var(F_j(\alpha_j))$ for all $i,j\leq n$.
$M$ is logically equivalent to the sequent
\[
  M'\qequiv \forall x_1( F_1(\alpha_1) \impl F_1(x_1)),
  \ldots, \forall x_n ( F_n(\alpha_n) \impl F_n(x_n)),
\forall x F(x) \seq
\]
which has Herbrand sequents of the form
\[
H' \qequiv  ( F_1(\alpha_1) \impl F_1(s_1) )_{s_1\in S_1},
\ldots,  ( F_n(\alpha_n) \impl F_n(s_n) )_{s_n\in S_n},
(F(s))_{s\in S} \seq.
\]
$H'$ is logically equivalent to the sequent
\[
H \qequiv  F_1(\alpha_1) \impl \bAND_{s_1\in S_1} F_1(s_1),
\ldots,  F_n(\alpha_n) \impl \bAND_{s_n\in S_n}F_n(s_n),
(F(s))_{s\in S} \seq.
\]
Hence we may identify $M$ and $M'$ and $H$ and $H'$ to be able
to talk about sequents $M$ and their Herbrand sequents $H$.

%
%
We are now ready to describe our algorithm.
We fix a Herbrand-sequent $H=\bAND_{t\in T} F\unsubst{x}{t}\seq $ of our
previously fixed sequent $S$ such that $T=L(G)$ for the grammar
$G=U\circ_{\alpha_1}S_1\cdots\circ_{\alpha_n}S_n$.
\begin{definition}\label{def:inter_solution}
We define a {\em $k$'th intermediary solution} to be a valid sequent
of the form
\[
F_n(\alpha_n)\impl\bAND_{s_n\in S_n} F_n(s_n),
\ldots, F_\ell(\alpha_\ell)\impl\bAND_{s_\ell\in S_\ell} F_\ell(s_\ell),
F(t)_{t \in T_\ell} \seq 
\]
where $\ell=n-k+1$ and $T_\ell=L(U \circ_{\alpha_1} S_1 \circ_{\alpha_2} \cdots \circ_{\alpha_{\ell-1}} S_{\ell-1})$.
\end{definition}
By the discussion above, a $k$'th intermediary solution is a Herbrand-sequent of
\[
  F_n(\alpha_n)\impl \forall x_n F_n(x_n),
\ldots, F_\ell(\alpha_\ell)\impl \forall x_\ell F_\ell(x_\ell),
\forall x F(x) \seq 
\]
Clearly, the $0$'th intermediary solution is just the Herbrand-sequent
$ (F(t))_{t\in T}\seq$,
and an $n$'th intermediary solution is a valid sequent of the form
\[
  F_n(\alpha_n)\impl\bAND_{s_n\in S_n} F_n(s_n), \ldots, F_1(\alpha_1)\impl\bAND_{s_1\in S_1} F_1(s_1),
  (F(u))_{u \in U} \seq.
\]
which is actually an extended Herbrand-sequent of $S$.
In other words, the substitution $\unsubst{X_i}{\lambda \alpha_i.F_i(\alpha_i)}_{i=1}^n$ solves
the schematic extended Herbrand sequent induced by $S$ and $G$.
We show now how to obtain such an $n$'th intermediary solution by
iteration of the results on the introduction of a single cut.
As observed above, we have:
\begin{lemma}\label{lem:zero_solution}
There exists a $0$'th intermediary solution.
\end{lemma}
Further, the following holds:
\begin{lemma}\label{lem:iterate_solution}
If there exists a $k$'th intermediary solution $I_k$,
then there exists a $(k+1)$'th intermediary solution $I_{k+1}$.
\end{lemma}
\begin{proof}
We use the notation from Definition~\ref{def:inter_solution} to denote $I_k$.
It suffices to show that
\[
  D = (S_n,\ldots,S_\ell,T_{\ell-1})\circ_{\alpha_{\ell-1}} S_{\ell-1}
\] 
generates the terms of $I_k$, 
i.e.~$L(D)=(S_n,\ldots,S_\ell,T_\ell)$. Assume so, 
and consider the schematic extended Herbrand-sequent
\[
  F_n(\alpha_n) \impl \bAND_{s_n\in S_n} F_n(s_n) , \ldots, F_{\ell}(\alpha_{\ell})\impl \bAND_{s_{\ell}\in S_{\ell}} F_{\ell}(s_{\ell}), X(\alpha_{\ell-1})\impl \bAND_{s_{\ell-1}\in S_{\ell-1}} X(s_{\ell-1}),
  (F(t))_{t \in T_{\ell-1}}\seq.
\]
This is a schematic extended Herbrand sequent since $I_k$ is valid and
$D$ generates its terms.
Then by 
Lemma~\ref{lem.can_sol_gen} there exists a solution $\sigma=\{\lambda \alpha_{\ell-1}. A\}$.
Putting $F_{\ell-1}=A$ we have that
\[
F_n(\alpha_n) \impl \bAND_{s_n\in S_n} F_n(s_n) , \ldots, F_{\ell-1}(\alpha_{\ell-1})\impl \bAND_{s_{\ell-1}\in S_{\ell-1}} F_{\ell-1}(s_{\ell-1}),
(F(t))_{t \in T_{\ell-1}}\seq
\]
is valid, which means that a $(k+1)$'th intermediary solution exists.
We verify that indeed $L(D)=(S_n,\ldots,S_\ell,T_\ell)$:
\begin{itemize}
  \item For $n\geq i\geq \ell$, $L(S_i\circ_{\alpha_{\ell-1}}S_{\ell-1})=S_i$ 
    since $\alpha_{\ell-1}=\alpha_{n-k}\notin V(S_i)$ by the definition of grammar.
  \item $L(T_{\ell-1}\circ_{\alpha_{\ell-1}}S_{\ell-1})=L((U \circ_{\alpha_1} S_1 \circ_{\alpha_2} \cdots \circ_{\alpha_{n-k-1}} S_{n-k-1})\circ_{\alpha_{n-k}}S_{n-k})=T_\ell$.
\end{itemize}
\end{proof}
Hence there exists an $n$'th intermediary solution which
yields a solution for the $n$-cut introduction problem. 
Before describing the algorithm $\SFn{\Cons}$, we note
that a Herbrand sequent $s$ and a grammar $G$ for its termset
induce a  schematic Herbrand sequent in a canonical way (where the number
of variables $X_i$ depends on the grammar), we denote
this schematic Herbrand sequent by $\SHS(s,G)$. In particular,
if $G$ is of the form $(U_1,\ldots,U_n)\circ_\alpha V$, then
$\SHS(s,G)$ contains exactly one schematic variable $X_1$.
\begin{algorithm}
\caption{\SFn{\Cons}}
\label{alg:simp_sol_n}
\begin{algorithmic}
\Function{\SFn{\Cons}}{$I_k$: $k$'th intermediary solution in CNF}
  \If{ $k=n$ }
    \State \Return $I_k$
  \EndIf
  \State $\ell \gets n-k+1$
  \State $C \gets $ canonical solution of $\SHS(I_k, (S_n,\ldots,S_\ell,T_{\ell-1})\circ_{\alpha_{\ell-1}} S_{\ell-1})$
  \State $F_{\ell-1} \gets \SF{\Cons}(C)$
  \State $I_{k+1} \gets k+1$'th intermediary solution based on $I_k$ and $F_{\ell-1}$
  \State \Return $\SFn{\Cons}(I_{k+1})$
\EndFunction
\end{algorithmic}
\end{algorithm}
The above results (note that by the proof of Lemma~\ref{lem:iterate_solution},
{\em any} $k$'th intermediary solution can be used
in the iteration) and observations, together with Theorem~\ref{thm:solution_finding}, entail
the following.
\begin{theorem}\label{thm:sfn_main}
  Let $\Cons$ be a consequence generator. If $F\in\SFn{\Cons}(H)$ then $F$ is an
  extended Herbrand-sequent based on $G$ (in particular, $F$ gives rise to a solution
  to $\SHS(H,G)$).
  If $\Cons$ is well-founded, then $\SFn{\Cons}(H)$ terminates.
\end{theorem}
Using forgetful resolution, we obtain the concrete algorithm $\CI$, which we
illustrate by an example.
\begin{example}\label{ex:improve_multiple}
Consider the example from Section~\ref{sec.motex} for $n=3$. That
is, we consider the sequent
\[
Pa, \forall x\, (Px\impl Pfx) \seq Pf^{8}a
\]
which has a Herbrand sequent $H$ with terms $T=\{a,fa,f^2a,\ldots,f^7a\}$.
$T$ is generated by the grammar $U\circ_{\alpha_1} S_1 \circ_{\alpha_2} S_2=
\{\alpha_1,f\alpha_1\}\circ\{\alpha_2,f^2\alpha_2\}\circ\{a,f^4a\}$.

From $H$, being the $0$'th intermediary solution,
we compute the first intermediary solution. To this end we
consider the grammar $L(U\circ S_1)\circ S_2=\{\alpha_2,f\alpha_2,f^2\alpha_2,f^3\alpha_2\}\circ \{a,f^4a\}$. 
This grammar leads to the schematic extended Herbrand sequent
\[
X\alpha_2 \impl (Xa \land Xf^4a), Pa, P\alpha_2\impl Pf\alpha_2,\ldots,Pf^3\alpha_2 \impl Pf^4\alpha_2 \seq Pf^{8}a.
\]
The canonical solution is $Pa\land (P\alpha_2\impl Pf\alpha_2 )\land \ldots \land (Pf^3\alpha_2 \impl Pf^4\alpha_2) \land \neg Pf^{8}a$ which,
when subjected to the simplification procedure for the 1-cut introduction problem, becomes
the solution
\[\neg P\alpha_2 \lor Pf^4\alpha_2.\]
Hence our first intermediary solution is
\[
  I_1= (\neg P\alpha_2 \lor Pf^4\alpha_2) \impl \bigwedge_{s\in S_2} (\neg Ps \lor Pf^4s), Pa, P\alpha_2\impl Pf\alpha_2,\ldots,Pf^3\alpha_2 \impl Pf^4\alpha_2 \seq Pf^{8}a
\]
which has the terms $(S_2,L(U\circ S_1))$ and is a Herbrand-sequent of
\[
  (\neg P\alpha_2 \lor Pf^4\alpha_2) \impl \forall x (\neg Px \lor Pf^4x), Pa, \forall x (Px\impl Pfx) \seq Pf^{8}a.
\]

We iterate the procedure to obtain the second intermediary solution. We consider
the grammar $(S_2, U)\circ_{\alpha_1} S_1=(\{a,f^4a\}, \{\alpha_1,f\alpha_1\})\circ \{\alpha_2,f^2\alpha_2\}$ of the terms of $I_1$. We obtain the schematic
extended Herbrand sequent
\[
\begin{array}{l}
X\alpha_1 \impl (X\alpha_2 \land Xf^2\alpha_2), Pa, (\neg P\alpha_2 \lor Pf^4\alpha_2) \impl \bigwedge_{s\in S_2} (\neg Ps \lor Pfs),\\
P\alpha_1\impl Pf\alpha_1,Pf\alpha_1 \impl Pf^2\alpha_1 \seq Pf^{8}a
\end{array}
\]
which has the canonical solution
\[
Pa \land ((\neg P\alpha_2 \lor Pf^4\alpha_2) \impl \bigwedge_{s\in S_2} (\neg Ps \lor Pfs)) \land
(P\alpha_1\impl Pf\alpha_1)\land(Pf\alpha_1 \impl Pf^2\alpha_1) \land \neg Pf^{8}a
\]
which simplifies to
\[
\neg P\alpha_1 \lor Pf^2\alpha_1.
\]
We finally obtain as the second intermediary solution the valid sequent
\[
\begin{array}{ll}
  I_2 = &(\neg P\alpha_1 \lor Pf^2\alpha_1) \impl \bigwedge_{s_1 \in S_1}(\neg Ps_1 \lor Pf^2s_1),\\
  &(\neg P\alpha_2 \lor Pf^4\alpha_2) \impl \bigwedge_{s_2 \in S_2}(\neg Ps_2 \lor Pf^4s_2),\\
  & Pa,
(Pu \impl Pfu)_{u\in U} \seq Pf^8a
\end{array}
\]
which is an extended Herbrand-sequent of $s$ and induces
the cut-formulas $\forall x(\neg Px \lor Pf^2x)$ and $\forall x(\neg Px \lor Pf^4x)$.
This can also be seen from the sequent for which $I_2$ is a Herbrand sequent:
\[
  \begin{array}{l}
(\neg P\alpha_1 \lor Pf^2\alpha_1) \impl \forall x_1 (\neg Px_1 \lor Pf^2x_1), 
(\neg P\alpha_2 \lor Pf^4\alpha_2) \impl \forall x_2 (\neg Px_2 \lor Pf^4x_2),\\
Pa, \forall x(Px \impl Pfx) \seq Pf^8a.
\end{array}
\]
\end{example}
%
%

\section{The Method CI}
\label{sec.method_CI}

In Section~\ref{sec.pth_instruct} we have shown that, for any rigid acyclic tree grammar $G$ generating the set of Herbrand terms $T$ of a cut-free proof
$\varphi$ of a sequent $S$, there exists a solution of the corresponding schematic extended Herbrand sequent $H$, the canonical solution. This canonical solution
yields cut formulas and an extended Herbrand sequent $H^*$ of $S$. By Theorem~\ref{thm.proof_extHseq}  we can construct a proof $\varphi^*$ with
cuts from $H^*$. In Section~\ref{sec.improving} we have presented several techniques to reduce the length of $\varphi^*$ (i.e. the logical complexity) under
preservation of the quantifier complexity. In combining all these transformations in a systematic way we obtain a
nondeterministic algorithm $\cutintro{A}{\Cons}$ where $A$ is an algorithm computing a minimal grammar; we may choose the algorithm {\DA} in 
Section~\ref{sec:computation_of_grammar} or the straightforward nondeterministic
guessing algorithm {\GG}. $\Cons$ is one of the consequence generators defined
Section~\ref{sec.improving}; e.g., we may choose $\Cons = \dedClosOp$ or $\Cons
= \FResOp$.

We now define $\cutintro{A}{\Cons}$:

\begin{itemize}
\item[] Input: a cut-free proof $\varphi$ of a $\Sigma_1$-sequent $S$.
\item[(1)] Compute the set of Herbrand terms $T$ of $\varphi$.
\item[(2)] Compute a minimal rigid acyclic tree grammar $G$ with $L(G) = T$ by $A$.
\item[(3)] Construct the canonical solution corresponding to $G$.
\item[(4)] Improve the canonical solution by $\SFn{\Cons}$.
\item[(5)] Construct the proof with the computed cut-formulas.
\end{itemize}
Note that, also without step (4), we obtain a full cut-introduction procedure.

\subsection{Inverting Cut-Elimination}

In this section we show that the method $\mathrm{CI}$ is complete by proving that
-- in a suitable sense -- it constitutes an inversion of Gentzen's procedure for
cut-elimination. To that aim we will rely on the results
of~\cite{Hetzl12Herbrand,HetzlXXHerbrand} which show that the language of a
grammar is a strong invariant of cut-elimination. To be more precise we quickly
repeat some notions from~\cite{Hetzl12Herbrand,HetzlXXHerbrand} here, for full
details the interested reader is referred to these papers.

\begin{definition}
We denote with $\cred$ the cut-reduction relation defined by
allowing the application
of the standard reduction rules without any strategy-restriction. The
standard reductions include rules such as e.g.
\[
\begin{array}{c}
\infer[\mathrm{cut}]{\Gamma, \Pi \seq \Delta, \Lambda}{
  \infer[\forall_\mathrm{r}]{\Gamma \seq \Delta, \forall x\, A}{
    \deduce{\Gamma \seq \Delta, A\unsubst{x}{\alpha}}{(\pi_1)}
  }
  &
  \infer[\forall_\mathrm{l}]{\forall x\, A, \Pi \seq \Lambda}{
    \deduce{A\unsubst{x}{t},\Pi\seq\Lambda}{(\pi_2)}
  }
}
\end{array}
\quad
\mapsto
\quad
\begin{array}{c}
\infer[\mathrm{cut}]{\Gamma, \Pi \seq \Delta, \Lambda}{
  \deduce{\Gamma \seq \Delta, A\unsubst{x}{t}}{(\pi_1\unsubst{\alpha}{t})}
  &
  \deduce{A\unsubst{x}{t},\Pi\seq\Lambda}{(\pi_2)}
}
\end{array}\ ,
\]
see~\cite[Figure 1]{Hetzl12Herbrand} for the complete list. With $\credne$
we denote the {\em non-erasing part} of $\cred$, i.e.\ we disallow application
of the reduction rule
\[
\begin{array}{c}
\infer[\mathrm{cut}]{\Gamma, \Pi \seq \Delta, \Lambda}{
  \infer[\mathrm{w}_\mathrm{r}]{\Gamma \seq \Delta, A}{
    \deduce{\Gamma \seq \Delta}{(\pi_1)}
  }
  &
  \deduce{A, \Pi\seq\Lambda}{(\pi_2)}
}
\end{array}
\quad
\mapsto
\quad
\begin{array}{c}
\infer[\mathrm{w}^*]{\Gamma, \Pi \seq \Delta, \Lambda}{
  \deduce{\Gamma \seq \Delta}{(\pi_1)}
}
\end{array}
\]
and its symmetric variant that removes a $\mathrm{w}_\mathrm{l}$-inference.
\end{definition}

\begin{definition}
A proof is called {\em simple} if every cut is of one of the following forms:
\[
\begin{array}{c}
\infer[\mathrm{cut}]{\Gamma, \Pi \seq \Delta, \Lambda}{
  \infer[\forall_\mathrm{r}]{\Gamma\seq\Delta, \forall x\, A}{
    \Gamma\seq\Delta, A\unsubst{x}{\alpha}
  }
  &
  \forall x\, A, \Pi \seq \Lambda
}
\end{array}
\quad
\mbox{or}
\quad
\begin{array}{c}
\infer[\mathrm{cut}]{\Gamma, \Pi \seq \Delta, \Lambda}{
  \Gamma\seq\Delta, \exists x\, A
  &
  \infer[\exists_\mathrm{l}]{\exists x\, A, \Pi \seq \Lambda}{
    A\unsubst{x}{\alpha}, \Pi \seq \Lambda
  }
}
\end{array}
\quad
\mbox{or}
\quad
\]
\[
\infer[\mathrm{cut}]{\Gamma, \Pi \seq \Delta, \Lambda}{
  \Gamma\seq\Delta, A
  &
  A, \Pi \seq \Lambda
}
\]
where $A$ is quantifier-free.
\end{definition}
Each proof with $\Pi_1\cup\Sigma_1$-cuts can be pruned to obtain a simple
proof by permuting $\forall_\mathrm{r}$- and $\exists_\mathrm{l}$-inferences
down and identifying their eigenvariables when needed. All of the reductions
of $\cred$ preserve simplicity with the exception of the following situation:
\[
\infer[\mathrm{cut}]{\Gamma, \Pi, \Sigma \seq \Delta, \Lambda, \Theta}{
  \Gamma\seq\Delta, A
  &
  \infer[\mathrm{cut}]{A, \Pi, \Sigma \seq \Lambda, \Theta}{
    \infer[\forall_\mathrm{r}]{A, \Pi \seq \Lambda, \forall x\, B}{
      A, \Pi \seq \Lambda, B\unsubst{x}{\alpha}
    }
    &
    \forall x\, B, \Sigma \seq \Theta
  }
}
\]
$\cred$
\[
\infer[\mathrm{cut}]{\Gamma, \Pi, \Sigma \seq \Delta, \Lambda, \Theta}{
  \infer[\mathrm{cut}]{\Gamma, \Pi \seq \Delta, \Lambda, \forall x\, B}{
    \Gamma\seq\Delta, A
    &
    \infer[\forall_\mathrm{r}]{A, \Pi \seq \Lambda, \forall x\, B}{
      A, \Pi \seq \Lambda, B\unsubst{x}{\alpha}
    }
  }
  &
  \forall x\, B, \Sigma \seq \Theta
}
\]
where the order of the two cuts is exchanged. We define a {\em reduction sequence
of simple proofs} as one where the above reduction is directly followed by
permuting down the $\forall_\mathrm{r}$-inference in order to arrive at:
\[
\infer[\mathrm{cut}]{\Gamma, \Pi, \Sigma \seq \Delta, \Lambda, \Theta}{
  \infer[\forall_\mathrm{r}]{\Gamma, \Pi \seq \Delta, \Lambda, \forall x\, B}{
    \infer[\mathrm{cut}]{\Gamma, \Pi \seq \Delta, \Lambda, B\unsubst{x}{\alpha}}{
      \Gamma\seq\Delta, A
      &
      A, \Pi \seq \Lambda, B\unsubst{x}{\alpha}
    }
  }
  &
  \forall x\, B, \Sigma \seq \Theta
}
\]
and symmetrically for the case of $\exists_\mathrm{l}$. We can now state the
part of the main result of~\cite{HetzlXXHerbrand} which is relevant for this paper:
\begin{theorem}\label{thm.grammar_invariance}
If $\pi\credne\pi^*$ is a cut-reduction sequence of simple proofs, then $\Lang(\Gram(\pi)) = \Lang(\Gram(\pi^*))$.
\end{theorem}
\begin{proof}
This is the second part of Theorem 7.2 in~\cite{HetzlXXHerbrand}.
\end{proof}

\begin{definition}
We write $\pi_1 \approx \pi_2$ if $\pi_1$ and $\pi_2$ have the same end-sequent and
$\Gram(\pi_1) = \Gram(\pi_2)$.
\end{definition}
Note that for simple $\pi_1,\pi_2$ we have $\pi_1 \approx \pi_2$ iff the schematic extended Herbrand-sequents
of $\pi_1$ and $\pi_2$ are identical up to renaming of the $X_i$. Furthermore,
note that $\approx$ is an equivalence relation. We write $[\pi]$ for the $\approx$-class
of $\pi$. Using the non-deterministic grammar guessing algorithm {\GG}, it is then
enough to invert $\credne$ up to $\approx$, more precisely:
\begin{theorem}
If $\pi\credne\pi^*$ is a cut-reduction sequence of simple proofs and $\pi^*$ is cut-free,
then $[\pi] \in \nicefrac{\cutintro{{\GG}}{\Cons}(\pi^*)}{\approx}$ for any well-founded consequence
generator $\Cons$.
\end{theorem}
\begin{proof}
By definition of the Herbrand-sequent we have
\[
H(\pi^*) =  \{ F_i\unsubst{\bar{x}}{\bar{t}} \mid 1\leq i \leq p, f_i(\bar{t}) \in \Lang(\Gram(\pi^*)) \}
\seq \{ F_i\unsubst{\bar{x}}{\bar{t}} \mid p < i \leq q, f_i(\bar{t}) \in \Lang(\Gram(\pi^*)) \}
\]
and so from Theorem~\ref{thm.grammar_invariance} we obtain
\[
H(\pi^*) =  \{ F_i\unsubst{\bar{x}}{\bar{t}} \mid 1\leq i \leq p, f_i(\bar{t}) \in \Lang(\Gram(\pi)) \}
\seq \{ F_i\unsubst{\bar{x}}{\bar{t}} \mid p < i \leq q, f_i(\bar{t}) \in \Lang(\Gram(\pi)) \}.
\]
Hence $\Gram(\pi)$ can be obtained from {\GG} applied to $\Hseq(\pi^*)$. Let $\psi$ be the proof obtained from
the canonical solution of $\Gram(\pi)$, then $\psi \in \cutintro{{\GG}}{\Cons}(\pi^*)$ (which is well defined
by Theorem~\ref{thm:sfn_main}) and as $\Gram(\psi) = \Gram(\pi)$
we have $\psi\approx\pi$ and therefore $[\pi] \in \nicefrac{\cutintro{{\GG}}{\Cons}(\pi^*)}{\approx}$.
\end{proof}

\subsection{Proof Compression}
We will prove in this section that application of CI to a sequence of cut-free
proofs $\varrho_n$ of sequents $S_n$ can result in an exponential compression of
$\varrho_n$. But we should take care with respect to which {\em complexity measure} the proofs
are compressed. In fact, the steps (1)-(3) and (5) yield proofs $\chi_n$ of
$S_n$ with cuts such that $|\varrho_n|_q$ is exponential in $|\chi_n|_q$, resulting in an
exponential compression of {\em quantifier complexity}. But these steps alone do
not yield an exponential compression of {\em proof length} (taking into account
{\em all} logical inferences). The compression of proof length is then achieved
by using step (4) of the algorithm; step (5) then yields a sequence of proofs
$\varphi_n$ of the sequents $S_n$ with cut s.t. $|\varrho_n|$ is exponential in
$|\varphi_n|$. 

As input we take a sequence of shortest cut-free proofs $\varrho_n$ of the sequents
$$s_n\colon Pa, \forall  x(Px \impl Pfx) \seq Pf^{2^{n+1}}a.$$
from Section~\ref{sec.motex}

\subsubsection{Exponential compression of quantifier-complexity}

We apply step (1) to $\varrho_n$ and obtain a sequence of the corresponding minimal Herbrand sequents
$$s'_n\colon Pa, (Ps \impl Pfs)_{s \in T_n} \seq Pf^{2^{n+1}}a,$$
where
$$T_n = \{a,fa,\ldots,f^{2^{n+1}-1}a\}.$$
Note that the quantifier complexity of any cut-free proof of $S_n$ is $ \geq 2^{n+1}$ and thus exponential in $n$.\\[1ex]
We continue with step (2) using {\DA}: Let $n$ be a fixed (but arbitrary) number
$>0$. The sets $T_n$ can be generated by the grammars
$$G_n\colon \{\alpha_1,f\alpha_1\}\circ_{\alpha_1} \{\alpha_2,f^2\alpha_2\} \circ_{\alpha_2} \cdots \circ \{\alpha_n,f^{2^{n-1}}\alpha_n\} \circ_{\alpha_n} \{a, f^{2^n}a\}.$$
where the $\alpha_i$ are variables and $a$ is a constant symbol.

\smallskip

For describing the
steps of the grammar computation algorithm {\DA} we define
\begin{eqnarray*}
T_k &=& \{\alpha_{k+1},\ldots,f^{2^{k+1}-1}\alpha_{k+1}\} \mbox{ for } k=0,\ldots,n \mbox{ and } \alpha_{n+1}=a,\\
S_k &=& \{\alpha_{k+1},f^{2^k}\alpha_{k+1}\} \mbox{ for } k=0,\ldots,n \mbox{ and } \alpha_{n+1}=a.
\end{eqnarray*}

Note that $S_0 = T_0 = \{\alpha_1,f\alpha_1\}$.

\smallskip

We have assumed that $n >0$. So {\DA} starts with $T_n$ and computes the
 grammar $T_{n-1} \circ S_n$. Assume now that the grammar
$$T_{n-k} \circ S_{n-k+1} \circ \cdots \circ_{\alpha_n} S_n$$
is already computed. If $k=n$ we are done; otherwise we have
$$T_{n-k} = \{\alpha_{n-k+1},\ldots,f^{2^{n-k+1}-1}\alpha_{n-k+1}\}$$
which is decomposed via {\DA} into $T_{n-k-1} \circ S_{n-k}$. Putting things together we obtain the grammar
$$T_{n-k-1} \circ S_{n-k} \circ S_{n-k+1} \circ \cdots \circ S_n.$$
Thus, eventually, we obtain the grammar
$$S_0 \circ \cdots \circ S_n.$$
for $T_n$ which is just $G_n$.\\[1ex]
We now move to step (3) and compute the canonical solution of the schematic extended Herbrand sequent $H_n$ corresponding to $G_n$ where
$$H_n = Pa,(Ps \impl Pfs)_{s \in S_0},X_1\alpha_1 \impl \bAND_{s \in S_1}X_1s, \ldots,X_n\alpha_n \impl \bAND_{s \in S_n}X_n s \seq Pf^{2^{n+1}}a.$$
The canonical solution of $H_n$ is
$$ \theta_n\colon  [X_1 \ass \lambda \alpha_1.C_1 \mid, \ldots, X_n \ass \lambda \alpha_n.C_n]$$ where
\begin{eqnarray*}
F &=& Pa \land (Px \impl Pfx) \land \neg Pf^{2^{n+1}}a,\\
C_1 &=& \bAND_{s \in S_0}F\{x \ass s\},\\
C_{i+1} &=& \bAND_{s \in S_i}C_i\{\alpha_i \ass s\} \mbox{ for } 1 \leq i <n.
\end{eqnarray*}
Now let $H_n^*\colon H_n\theta_n$ be the corresponding extended Herbrand sequent. Then $|H_n^*| = 2(n+1)$. By Theorem~\ref{thm.proof_extHseq} we can construct proofs $\chi_n$
of $S_n$ with $|\chi_n|_q = |H_n^*| = 2(n+1)$. As $|\varrho_n|_q \geq 2^{n+1}$ we obtain
$$|\varrho_n|_q \geq 2^{|\chi_n|_q/2},$$
and so $|\varrho_n|_q$ is exponential in $|\chi_n|_q$.

\subsubsection{Exponential compression of proof length}

Let
$$F_k = (\neg P\alpha_k \lor Pf^{2^k}\alpha_k) \impl \bAND_{s \in S_k} (\neg Ps \lor Pf^{2^k}s)$$
for $k=1,\ldots,n$.

\smallskip

We will prove that using phase (4) of $\cutintro{\DA}{\FResOp}$ (we apply forgetful resolution) we can obtain the extended Herbrand sequent
$$s^*_n\colon F_n,\ldots,F_1, Pa, (Ps \impl Pfs)_{s \in S_0} \seq Pf^{2^{n+1}}a$$
from $s_n$ and the grammar $S_0 \circ \cdots \circ S_n$ for $T_n$.

For the proof we need two auxiliary lemmas.

\begin{lemma}\label{le.speedup1}
Let $k \geq 1$. Then the formula $\neg P\alpha_{n-k+1} \lor Pf^{2^{n-k+1}}\alpha_{n-k+1}$ is derivable by forgetful resolution from $\bAND_{s \in T_{n-k}}(Ps \impl Pfs)$.
\end{lemma}
\begin{proof}
Let $G_k = \bAND_{s \in T_{n-k}}(Ps \impl Pfs)$. By definition of $T_{n-k}$ $G_k$ is the formula
$$(P\alpha_{n-k+1} \impl Pf\alpha_{n-k+1}) \land \cdots \land (Pf^{2^{n-k+1}-1}\alpha_{n-k+1} \impl Pf^{2^{n-k+1}}\alpha_{n-k+1}).$$
Its conjunctive normal form $G'_k$ is
$$(\neg P\alpha_{n-k+1} \lor Pf\alpha_{n-k+1}) \land \cdots \land (\neg Pf^{2^{n-k+1}-1}\alpha_{n-k+1} \lor Pf^{2^{n-k+1}}\alpha_{n-k+1}).$$
By setting $x_i = Pf^i\alpha_{n-k+1}$ we obtain
$$G'_k = \bAND^{2^{n-k+1}-1}_{i=0}(\neg x_i \lor x_{i+1}).$$
The first step of forgetful resolution (resolving the first two clauses) gives us the CNF:
$$(\neg x_0 \lor x_2)\land \bAND^{2^{n-k+1}-1}_{i=2}(\neg x_i \lor x_{i+1}).$$
By repeating resolving the first two clauses we eventually obtain $\neg x_0 \lor x_{2^{n-k+1}}$ which is just $\neg P\alpha_{n-k+1} \lor Pf^{2^{n-k+1}}\alpha_{n-k+1}$.
\end{proof}

\begin{lemma}\label{le.speedup2}
Let $k<n$ and
$$s\colon X\alpha_{n-k} \impl \bAND_{s \in S_{n-k}} Xs, F_n,\ldots,F_{n-k+1},Pa,(Ps \impl Pfs)_{s \in T_{n-k-1}} \seq Pf^{2^{n+1}}a$$
be a schematic sequent. Then the substitution
$$[X \ass \lambda x.(\neg Px \lor Pf^{2^{n-k}}x)]$$
is a solution of $s$.
\end{lemma}
\begin{proof}
We prove that $s'\colon s[X \ass \lambda x.(\neg Px \lor Pf^{2^{n-k}})]$ is a valid sequent. $s'$ is of the form
$$F_n,\ldots,F_{n-k+1},F_{n-k}, Pa,(Ps \impl Pfs)_{s \in T_{n-k-1}} \seq Pf^{2^{n+1}}a$$
where
$$F_{n-k} = (\neg P\alpha_{n-k} \lor Pf^{2^{n-k}}\alpha_{n-k}) \impl \bAND_{s \in S_{n-k}}(\neg Ps \lor Pf^{2^{n-k}}s).$$
By Lemma~\ref{le.speedup1} we have $(Ps \impl Pfs)_{s \in T_{n-k-1}} \vdash \neg P\alpha_{n-k} \lor Pf^{2^{n-k}}\alpha_{n-k}$ by forgetful resolution. By modus ponens
with $F_{n-k}$ we obtain
$$(\neg P\alpha_{n-k+1} \lor Pf^{2^{n-k}}\alpha_{n-k+1}) \land (\neg Pf^{2^{n-k}}\alpha_{n-k+1} \lor Pf^{2^{n-k+1}}\alpha_{n-k+1})$$
from which, by resolution, we obtain
$$C_{n-k+1}\colon \neg P\alpha_{n-k+1} \lor Pf^{2^{n-k+1}}\alpha_{n-k+1}.$$
If $k=0$ we obtain $C_{n+1}$ and the sequent
$$(*)\ Pa, Pa \impl Pf^{2^{n+1}}a, \ldots, \seq Pf^{2^{n+1}}a$$
which is valid.

\smallskip

If $k>0$ then by modus ponens on $C_{n-k+1}$ and $F_{n-k+1}$ and resolving the result we obtain $C_{n-k+2}$, and so forth. Eventually we obtain $C_{n+1}$ and $(*)$.
\end{proof}

\begin{proposition}\label{prop.speedup1} Let $k \leq n$. Then the
solution-finding algorithm $\SFn{\FResOp}$ constructs the intermediary solution
$$s^+_k\colon F_{n},\ldots,F_{n-k+1}, Pa, (Ps \impl Pfs)_{s \in T_{n-k}} \seq Pf^{2^{n+1}}a.$$
\end{proposition}
\begin{proof}

For $k=0$ we just have the (valid) input sequent $s_n$ and the Herbrand
instances from $T_n$, trivially constructed by $\SFn{\FResOp}$.

\smallskip

Assume that $0<k<n$ and $\SFn{\FResOp}$ has constructed the intermediary solution
$$s^+_k\colon F_n,\ldots,F_{n-k+1},Pa, (Ps \impl Pfs)_{s \in T_{n-k}}  \seq Pf^{2^{n+1}}a$$
with Herbrand terms $T_{n-k}$. By definition of the $S_i,T_i$ we have that $T_{n-k-1} \circ S_{n-k}$ is a grammar for $T_{n-k}$. Via this grammar we obtain the
schematic extended Herbrand sequent
$$s\colon X\alpha_{n-k} \impl \bAND_{s \in S_{n-k}}Xs,F_n,\ldots,F_{n-k+1},Pa, (Ps \impl Pfs)_{s \in T_{n-k-1}} \seq Pf^{2^{n+1}}a.$$
As $V(T_{n-k-1}) = \{\alpha_{n-k}\}$ and $\alpha_{n-k}$ does not occur in
$F_n,\ldots,F_{n-k+1}$, the improved canonical solution, according to Lemma~\ref{lem:clean_ground},
constructed by $\SFn{\FResOp}$ is
$$[X \ass \lambda x.\bAND_{s \in T_{n-k-1}}(Ps \impl Pfs)[\alpha_{n-k} \ass x]].$$
By Lemma~\ref{le.speedup1} forgetful resolution constructs the formula
$$\neg Px \lor Pf^{2^{n-k}}x \mbox{ from } \bAND_{s \in T_{n-k-1}}(Ps \impl Pfs)\unsubst{\alpha_{n-k}}{x}.$$
By using Lemma~\ref{le.speedup2} we conclude that
$$[X \ass \lambda x.(\neg Px \lor Pf^{2^{n-k}}x)]$$
is a solution of $s$ yielding the sequent
$$F_n,\ldots,F_{n-k+1},F_{n-k}, Pa, (Ps \impl Pfs)_{s \in T_{n-k-1}} \seq Pf^{2^{n+1}}a.$$

\end{proof}

\begin{corollary}\label{coro.speedup1}
Let $s^*_n$ be the schematic extended Herbrand sequent
$$X_n\alpha_n \impl \bAND_{s \in S_n}X_n s, \ldots, X_1\alpha_1 \impl \bAND_{s \in S_1}X_1 s, Pa, (Ps \impl Pfs)_{s \in S_0} \seq Pf^{2^{n+1}}a$$
corresponding to $s_n$ and the grammar $G_n$. Then $\SFn{\FResOp}$ constructs a solution $\vartheta$ for $s^*_n$ where
$$\vartheta = [X_1 \ass \lambda x.(\neg Px \lor Pf^2x), \ldots, X_n \ass \lambda x.(\neg Px \lor Pf^{2^n}x)].$$

\end{corollary}

\begin{proof}
Obvious by Proposition~\ref{prop.speedup1} and by definition of the $F_i$.
\end{proof}

%
New we construct the proofs.
We define a sequence of proofs $\varphi_n$ (via $\psi_n$) in the following way:\\[1ex]
$\psi_0$ is a cut-free proof of the sequent $\all x(Px \impl Pfx) \seq \all x(Px \impl Pf^2x)$
and\\[1ex]
$\psi_{k+1} =$
\[
\infer[\cut]{\all x(Px \impl Pfx) \seq \all x(\neg Px \lor Pf^{2^{k+2}}x)}
 { \deduce{\all x(Px \impl Pfx) \seq \all x(\neg Px \lor Pf^{2^{k+1}}x)}{(\psi_k)}
   &
   \deduce{\all x(\neg Px \lor Pf^{2^{k+1}}x) \seq \all x(\neg Px \lor Pf^{2^{k+2}}x)}{(\chi_{k+1})}
 }
\]
where the eigenvariables of $\psi_k$ are $\alpha_1,\ldots,\alpha_{k+1}$ and $\chi_{k+1}$ is a cut-free proof of constant length
with eigenvariable $\alpha_{k+2}$ and $\foralll$-instances $\{\alpha_{k+2}, f^{2^{k+1}}\alpha_{k+2}\}$.

\smallskip

We can now define $\varphi_n$ for $n > 0$:\\[1ex]
$\varphi_n=$
\[
\infer[\cut]{Pa, \all x(Px \impl Pfx) \seq Pf^{2^{n+1}}a}
 { \deduce{\all x(Px \impl Pfx) \seq \all x(\neg Px \lor Pf^{2^n}x)}{(\psi_{n-1})}
   &
   \deduce{ Pa, \all x(\neg Px \lor Pf^{2^n}x) \seq Pf^{2^{n+1}}a}{(\sigma_n)}
 }
\]
where $\sigma_n$ is a constant length proof with Herbrand terms $\{a, f^{2^n}a\}$.

\begin{proposition}\label{prop.speedup2}
Given a proof of $s_n$ for $n>0$ with minimal Herbrand sequent $s'_n$ (and instances $T_n$) and the grammar $S_0 \circ \ldots \circ S_n$ for $T_n$, the
proof construction algorithm {\PCA} constructs the proof $\varphi_n$.
\end{proposition}
\begin{proof}
By Corollary~\ref{coro.speedup1} CI constructs the extended Herbrand sequent
$$F_n,\ldots,F_1, Pa, (Ps \impl Pfs)_{s \in S_0} \seq Pf^{2^{n+1}}a$$
where
\begin{eqnarray*}
F_1 &=& (\neg P\alpha_1 \lor Pf^2\alpha_1) \impl \bAND_{s \in S_1} (\neg Ps \lor Pf^2s)\\
\ldots & & \ldots\\
\ldots & & \ldots\\
F_n &=& (\neg P\alpha_n \lor Pf^{2^n}\alpha_k) \impl \bAND_{s \in S_n} (\neg Ps \lor Pf^{2^n}s)
\end{eqnarray*}
From the $F_i$ we read off the cut-formulas $\all x(\neg Px \lor Pf^{2^i}x)$
with eigenvariable substitution $\unsubst{x}{\alpha_i}$ for the left side of the
cut and substitutions $\unsubst{x}{\alpha_{i+1}}, \unsubst{x}{f^{2^i}\alpha_{i+1}}$
for the right side. Note that these are exactly the quantifier substitutions in
the proofs $\psi_i$. For $F_n$ the corresponding substitutions are
$\unsubst{x}{\alpha_n}$ and $\unsubst{x}{a}, \unsubst{x}{f^{2^n}a}$, respectively, corresponding to the last cut in
$\varphi_n$.

\smallskip

Now taking $F_1$ we construct the proof $\psi_1$  via $\all x(Px \impl Pfx)$
(which occurs in the end-sequent) and the cut-formula $\neg Px \lor Pf^2x$ and
via the substitutions $\unsubst{x}{\alpha_1}$ for the eigenvariable of the cut,
$\{ \unsubst{x}{s} \mid s \in S_0 \}$ for the instantiations of
$\all x(Px \impl Pfx)$,  $\{ \unsubst{x}{s} \mid s \in S_1\}$ for the right side of the cut. This gives us all quantifier-inferences of $\psi_1$; it remains to order the inferences appropriately to obtain
$\psi_1$. Note that $Pa$ on the left hand side of the end-sequent is not needed for constructing the proof.

\smallskip

Having constructed $\psi_k$ ($k<n$) with the cut-formulas $\neg Px \lor Pf^2x$, \ldots, $\neg Px \lor Pf^{2^k}x$ we have processed the formulas $F_1,\ldots,F_k$. Now
$$F_{k+1} = (\neg P\alpha_{k+1} \lor Pf^{2^{k+1}}\alpha_{k+1}) \impl \bAND_{s \in S_{k+1}} (\neg Ps \lor Pf^{2^{k+1}}s)$$
We know by construction that the last eigenvariable substitution in $\psi_k$ is
$\unsubst{x}{\alpha_{k+1}}$ and the end-formula on the right hand side is $\all x(\neg
Px \lor Pf^{2^{k+1}}x)$, which by $F_{k+1}$ is the next cut-formula. The
instantiations for the right hand side of the cut are $\{ \unsubst{x}{s} \mid s \in S_{k+1}\}$. If
$k+1<n$ we have $F_{k+2}$, from which $\all x(\neg Px \lor Pf^{2^{k+2}}x)$ is the cut-formula, which is also the end-formula of $\psi_{k+1}$. With this information we
have all necessary quantifier substitutions and formulas to construct $\psi_{k+1}$. If $k+1 =n$ we have
$$F_n = (\neg P\alpha_n \lor Pf^{2^{n}}\alpha_n) \impl \bAND_{s \in \{a,f^{2^n}a\}} (\neg Ps \lor Pf^{2^n}s)$$
and thus $\{a,f^{2^n}a\}$ as substitutions of the right side of the cut. We also know the end sequent $s_n$. This information eventually yields the proof $\varphi_n$.
\end{proof}

\begin{theorem}\label{the.speedup}
Let $(\rho_n)_{(n \in \Nat)}$ be a sequence of shortest cut-free proofs of
$(s_n)_{(n \in \Nat)}$. Then $|\rho_n|> 2^n$ and there are constants $a$,$b$
s.t. the cut-introduction algorithm, applied to $\rho_n$, constructs a sequence
of proofs $\varphi_n$ of $s_n$ s.t. $|\varphi_n| \leq a*n + b$ for all $n \geq 2$.
\end{theorem}
\begin{proof}
We have $|\rho|>2^n$ for any cut-free proof $\rho$ of $s_n$ and, in particular, for a shortest proof $\rho_n$. The Herbrand instances of shortest proofs of
$s_n$ must be $T_n$, as $T_n$ is a minimal set of Herbrand terms for $s_n$, and
therefore the Herbrand instances for $\rho_n$ are just $T_n$. So the {\DA} constructs the
grammar $S_0 \circ \cdots \circ S_n$ of $T_n$. Then by Proposition~\ref{prop.speedup2} the algorithm constructs the sequence of proofs $\varphi_n$ of $s_n$ for $n \geq 2$.
The length of $\varphi_n$ is linear in $n$. In fact $|\varphi_n| = |\psi_{n-1}| + |\sigma_n| + 1$, where $|\sigma_n| = c$ for a constant $c$.

Moreover, $|\psi_0| = d$ for some constant $d$ and $|\psi_{k+1}| = |\psi_k| + |\chi_{k+1}| + 1$,
where $|\chi_{k+1}| = e$ for some constant $e$. So $|\psi_k| = d +
k*(e+1)$ and
$$|\varphi_n| = c+ 1 + d + (n-1)*(e+1).$$

\end{proof}

\section{Implementation and Experiments}
\label{sec.implementation_experiments}

The algorithm described in this paper was implemented in the
gapt-system\footnote{\url{http://www.logic.at/gapt/}} for introducing a
single $\Pi_1$-cut into a sequent calculus proof. Gapt is a framework for implementing proof
transformations written in the programming language Scala. It was initially
developed for eliminating cuts of proofs by using resolution (CERES), 
but it has proven to be general enough so that other transformations
could be implemented in it. In this section, we explain how to
run the cut-introduction algorithm to compress a proof.

In order to install gapt, you need to have a Java Runtime Environment
(JRE\footnote{\url{http://www.oracle.com/technetwork/java/javase/downloads/jre7-downloads-1880261.html}})
installed.
Then, go to \url{http://www.logic.at/gapt} and the \texttt{gapt-cli-1.4.zip}
file should be
available in the ``downloads'' section. After uncompressing this file, you
should see a directory where you can find the running script \texttt{cli.sh} and
a README file. To run the program, just execute this script.

Gapt opens in a Scala interactive shell (\texttt{scala>}) from where you
can run all the commands provided by the system. To see a list of them, type
\texttt{help}. The commands are separated in categories, and we are interested in
the ones listed under ``Cut-Introduction'' and ``Proof Examples''. In this list
you can see the types of the functions and a brief description of what they do.
Observe that there exist functions for each of the steps described in this
paper for the introduction of cuts, and there is also a ``cutIntro'' command
that does all steps automatically. Under ``Proof Examples'' there is a set of
functions that generate (minimal) cut-free proofs of some parametrized
end-sequents. For example, \texttt{LinearExampleProof(n)} will generate a
cut-free proof of the end-sequent $P0, \forall x (Px \impl Psx)
\rightarrow Ps^n0$ for some $n$.

We will use as input one of these
proofs generated by the system, namely, \texttt{LinearExampleProof(9)}. But the
user can also, for example, write his own proofs in
\texttt{hlk}\footnote{\url{http://www.logic.at/hlk/}} and input these files to
the system. Currently, there is also an effort to implement a parser for proofs
obtained from the TPTP and SMT-LIB problem libraries so that large scale experiments can be
carried out. Meanwhile, we use the motivating example (Section \ref{sec.motex}) for
a brief demonstration.

First of all, we instantiate the desired proof and store this in a variable:
{\small
\begin{alltt}
scala> val p = LinearExampleProof(9) 
\end{alltt}
}
You will see that a big string representing the proof is printed. Gapt also
contains a viewer for proofs and other elements \cite{ProoftoolUITP2012}. You can open it to view
a proof $p$ at any moment with the command \texttt{prooftool(p)}.
It is possible to
see some information about a proof on the command line by calling:
{\small
\begin{alltt}
scala> printProofStats(p)
------------- Statistics ---------------
Cuts: 0
Number of quantifier rules: 9
Number of rules: 28
Quantifier complexity: 9
----------------------------------------
\end{alltt}
}
Now we need to
extract the terms used to instantiate the $\forall$ quantifiers of the
end-sequent:
{\small
\begin{alltt}
scala> val ts = extractTerms(p)
\end{alltt}
}
The system indicates how many terms were extracted, which is nine for this case,
as expected. The next step consists in computing grammars that generate this
term set (Section \ref{sec:computation_of_grammar}):
{\small
\begin{alltt}
scala> val grms = computeGrammars(ts)
\end{alltt}
}
The number of grammars found is shown, 693 in this case. They are ordered
by size, and one can see the first ones by calling:
{\small
\begin{alltt}
scala> seeNFirstGrammars(grms, 5)
\end{alltt}
}
This will print on the screen the first 5 grammars, and we can choose which one
to use for compressing the proof, in this case we take the second one:
{\small
\begin{alltt}
scala> val g = grms(1)
\end{alltt}
}
Given the end-sequent of the proof and a grammar, the extended Herbrand sequent can be computed:
{\small
\begin{alltt}
scala> val ehs = generateExtendedHerbrandSequent(p.root, g)
\end{alltt}
}
As was shown in Lemma \ref{lem.can_sol}, the cut-introduction problem has a
canonical solution:
{\small
\begin{alltt}
scala> val cs = computeCanonicalSolution(p.root, g)
\end{alltt}
}
The canonical solution is then printed on the screen. The extended Herbrand
sequent generated previously has the canonical solution as
default, but this solution can be improved, as demonstrated in Section
\ref{sec.improving}. 
{\small
\begin{alltt}
scala> minimizeSolution(ehs)
\end{alltt}
}
The user can see then that the canonical solution: $\forall x.((P(x)\impl
P(s(x))) \wedge ((P(s(s(x))) \impl P(s(s(s(x))))) \wedge (P(s(x)) \impl
P(s(s(x))))))$ is transformed into a simpler one: $\forall x.(P(s(s(s(x)))) \vee
\neg P(x))$.
Finally, the proof with cut is constructed:
{\small
\begin{alltt}
scala> val fp = buildProofWithCut(ehs)
\end{alltt}
}
In order to compare this with the initial proof, one can again count
the number of rules:
{\small
\begin{alltt}
scala> printProofStats(fp)
------------- Statistics ---------------
Cuts: 1
Number of quantifier rules: 7
Number of rules: 25
Quantifier complexity: 6
----------------------------------------
\end{alltt}
}
We showed how to run the cut-introduction algorithm step by step. There is,
though, a command comprising all these steps:
{\small
\begin{alltt}
scala> val fp2 = cutIntro(p)
\end{alltt}
}
Regarding the choice of the grammar, this command will compute the proofs with
all minimal grammars, and will output the smallest proof (with respect to the
number of rules). Of course, there might be two grammars that generate equally small
proofs. In this case, any grammar/proof can be chosen as a solution.

Besides \texttt{LinearExampleProof}, other sequences of cut-free proofs 
--- similar in spirit, but technically different from \texttt{LinearExampleProof} --- are encoded
in the gapt-system. To emphasize the potential of our cut-introduction method, we
present in Table~\ref{tab:example_results} the results (i.e.~the generated cut-formulas 
and the compression ratio obtained by dividing the number of inferences of the generated
proof by the number of inferences of the input proof)
of applying the implementation to some instances of these sequences. 
All of the displayed proofs involving $=$ use a usual axiomatization of equality based
on reflexivity, symmetry, transitivity, and congruence.
Of course, these 
results do not constitute a systematic empirical
investigation, and evaluation of the method using larger data sets is a necessity. Such experiments
are left for future work and are discussed in the following section.
\begin{table}
\begin{tabular}{|llll|}
\hline
Name & End-sequent & Cut-formula & CR\\
\hline
\texttt{SquareDiagonal} & \parbox{6cm}{$P(0,0), \forall x, y. P(x,y) \impl P(s(x),y),$\\ $\forall x, y. P(x,y) \impl P(x,s(y))$\\ $\seq P(s^8(0),s^8(0))$} & $\forall x. P(x,x) \impl P(s^2(x),s^2(x))$ & $0.53$\\
\hline
\texttt{LinearEq} & $\forall x.f(x)=x \seq f^8(a)=a$ & $\forall x. x=a \impl f^2(x)=a$ & $0.54$\\
\hline
\texttt{SumOfOnes} & \parbox{6cm}{$\forall x. x+0 = x, \forall x, y. x+s(y)=s(x+y)$\\ $\seq 1+1+1+1+1+1=s^6(0)$} & $\forall x . x+s(0) = s(x)$ & $0.56$ \\
\hline
\texttt{SumOfOnesF} & \parbox{6cm}{$\forall x. x+0 = x, \forall x, y. x+s(y)=s(x+y)$\\ $f(0)=0, \forall x.f(s(x)) = f(x) + s(0)$\\ $\seq f(s^4(0)) =s^4(0)$} & $\forall x. f(s(x)) = s(f(x))$& $0.69$\\
\hline
\texttt{SumOfOnesF2} & \parbox{6cm}{$\forall x. x+0 = x, \forall x, y. x+s(y)=s(x+y)$\\ $f(0)=0, \forall x.f(s(x)) = f(x) + s(0)$\\ $\seq f(s^4(0)) =s^4(0)$} & $\forall x. f(x) = x \impl f(s(x)) = s(x)$& $0.45$\\
\hline
\end{tabular}
\caption{Initial experimental results}
\label{tab:example_results}
\end{table}

\section{Conclusion and Future Work}

We have described a method for the inversion of Gentzen's cut-elimination method
by the introduction of quantified cuts into an
existing proof. Our method is based on separating the problem into two phases:
first the minimization of a tree grammar and secondly: finding a solution of
a unification problem. This separation is based on proof-theoretic results
which makes the method computationally feasible as demonstrated by its implementation.

The work presented in this paper is only a first step and opens
up several important directions for future work: a straightforward extension of
this algorithm is to use blocks of quantifiers
in the cut-formulas. This ability is useful for obtaining additional abbreviations. 
It would require modifications in the computation of the $\Delta$-vector and, given
the length of the paper, we have decided to leave this feature to future work.
An equally obvious -- but less straightforward -- extension of this method is
to cover cut-formulas with quantifier alternations. As prerequisite for this work, the
extension of the connection between cut-elimination and tree languages
established in~\cite{Hetzl12Applying} to the corresponding class of formulas is necessary.

On the empirical side an important aspect of future work will be to assess
the abilities for compression of proofs produced by theorem provers: we
intend to carry out large scale experiments with proofs produced from the
TPTP-library~\cite{Sutcliffe09TPTP} and the SMT-LIB~\cite{Barret10Satisfiability}.
Additional features that we consider important for such applications are to include the
ability to work modulo simple theories and to systematically compute grammars
whose language is a superset of the given set of terms.

On the theoretical side, we could only scratch the surface of many questions
in this paper: What is the complexity of grammar minimization? What are
good exact algorithms? Can we find incompressible tree languages? And more generally:
study the complexity of cut-free proofs along the lines of measures such as automatic
complexity~\cite{Shallit01Automatic} and automaticity~\cite{Shallit96Automaticity}.
Also the unification problem poses a number of interesting theoretical
challenges: Is the general unification problem of monadic predicate variables
modulo propositional logic decidable? If yes, what is its complexity, what are good algorithms?
What is the structure of the solution space of unification problems induced by
cut-introduction? How can we navigate this structure systematically to find
solutions of minimal size? How can we prover lower bounds on the size of such
solutions?
\bibliography{references}
\bibliographystyle{plain}

\end{document}